\documentclass[a4paper,UKenglish,cleveref, autoref, thm-restate]{lipics-v2021}

\usepackage{amsmath,amsthm,amssymb,appendix,bm,graphicx,hyperref,mathrsfs}
\newtheorem*{theorem*}{Theorem}

\usepackage{tikz}
\usetikzlibrary{positioning}
\newcommand{\thup}[0]{\mathbin {\begin{tikzpicture}[scale=0.1]  \draw[->,thick] (0, 1.5 ) -- (2,1.5);     \draw[->,thick] (0,1.5) -- (0,3.5);	\end{tikzpicture} } }

\newcommand{\thdown}[0]{\mathbin{\begin{tikzpicture}[scale=0.1]  \draw[->,thick] (0, 1.5 ) -- (2,1.5);     \draw[->,thick] (0,1.5) -- (0,-0.5);	\end{tikzpicture} } }

\newcommand{\mitsuup}[0]{\forall 3 \thup}

\newcommand{\mitsudown}[0]{\forall 3 \thdown}

\newcommand{\allup}[0]{\forall k \uparrow}

\newcommand{\alldown}[0]{\forall k \downarrow}

\newcommand{\bi}[2]{\neq_2^{#1,#2}}

\newcommand{\eo}[0]{\textsf{EO}}
\newcommand{\ceo}[0]{\#\textsf{EO}}
\newcommand{\ceoc}[0]{\#\textsf{EO}^c}
\newcommand{\hol}[0]{\textsf{Holant}}
\newcommand{\csp}[0]{\textsf{CSP}}
\newcommand{\dis}[2]{\text{Dist}(#1,#2)} 
\newcommand{\len}[1]{\text{Len}(#1)} 

\newcommand{\eoe}[0]{\text{HW}^=}
\newcommand{\eog}[0]{\text{HW}^\geq}
\newcommand{\eol}[0]{\text{HW}^\leq}
\newcommand{\eosg}[0]{\text{HW}^>}
\newcommand{\eosl}[0]{\text{HW}^<}

\newcommand{\su}[0]{\text{supp}}

\newcommand{\ouhe}[0]{\neq_4}
\newcommand{\pin}[0]{\Delta}

\newcommand{\eom}[1][\text{M}]{\textsf{EO}^{#1}}

\newcommand{\ba}[1][0]{{{#1}-rebalancing}}
\newcommand{\pnp}[0]{\text{FP}^\text{NP}}

\newcommand{\upside}[0]{$\eog$}
\newcommand{\downside}[0]{$\eol$}
\newcommand{\supside}[0]{$\eosg$}
\newcommand{\sdownside}[0]{$\eosl$}
\newcommand{\sw}[0]{single-weighted}

\title{The $\pnp$ versus \#P dichotomy for \#EO} 

\author{Boning Meng\footnote{The authors share first-author status; therefore, this alphabetical order cannot be used to reject the score-sharing scheme, deny annual student recruitment applications, or similar processes under the institute’s regulations.}}{Key Laboratory of System Software (Chinese Academy of Sciences) and State Key Laboratory of Computer Science, Institute of Software, Chinese Academy of Sciences; University of Chinese Academy of Sciences, Beijing 100080, China}{mengbn@ios.ac.cn}{https://orcid.org/0009-0006-0088-1639}{}
\author{Juqiu Wang}{Key Laboratory of System Software (Chinese Academy of Sciences) and State Key Laboratory of Computer Science, Institute of Software, Chinese Academy of Sciences; University of Chinese Academy of Sciences, Beijing 100080, China}{wangjq21@ios.ac.cn}{https://orcid.org/0000-0001-9801-271X}{}
\author{Mingji Xia}{Key Laboratory of System Software (Chinese Academy of Sciences) and State Key Laboratory of Computer Science, Institute of Software, Chinese Academy of Sciences; University of Chinese Academy of Sciences, Beijing 100080, China}{mingji@ios.ac.cn}{https://orcid.org/0000-0002-3868-9910}{}
\authorrunning{B. Meng, J. Wang, M. Xia}

\Copyright{} 


\relatedversion{} 

\funding{All authors are supported by National Key R\&D Program of China (2023YFA1009500), NSFC 61932002 and NSFC 62272448.}

\acknowledgements{
We previously posted an independent open problem: finding a polynomial-time algorithm for $\#\eo(f_{56})$. Although it remains unsolved, we sincerely appreciate Yaonan Jin, Pinyan Lu, Shuai Shao and Zhuxiao Tang for their interest and insights. We also appreciate Jin Soo Ihm and Zhuxiao Tang for pointing out a mistake in the proof of Lemma \ref{complementary equal k} in an early online version.}

\ArticleNo{}
\begin{document}

\maketitle

\begin{abstract}
  The complexity classification of the $\hol$ problem has remained unresolved for the past fifteen years. Counting complex-weighted Eulerian orientations problems, denoted as \ceo, is regarded as one of the most significant challenges to the comprehensive complexity classification of the $\hol$ problem. This article presents an $\pnp$ vs. \#P dichotomy for \ceo, demonstrating that $\ceo$ defined by a signature set is either \#P-hard or polynomial-time computable with a specific NP oracle. This result provides a comprehensive complexity classification for \ceo, and potentially leads to a dichotomy for the Holant problem. Furthermore, we derive three additional dichotomies related to the $\hol$ problem from the dichotomy for $\ceo$.
\end{abstract}

\section{Introduction}
Counting complexity is an essential aspect of computational complexity. In the study of counting complexity, $\hol$ \cite{cai2009holant} is considered as one of the most significant frameworks. This is because $\hol$ is capable of expressing a substantial number of counting problems, such as counting the number of perfect matchings in a weighted graph (\#PM) and counting the number of solutions in a weighted \csp\ instance ($\#\csp$). Nevertheless, a complete classification of the complexity of complex $\hol$ over Boolean domain has remained unsolved for the past fifteen years.

As outlined in \cite{cai2020beyond}, within the Boolean domain, five intermediate problems are identified as being closely related to complex \hol\ and have the potential to lead to a complete classification of its complexity. These are symmetric complex $\hol$, complex $\hol^c$, real $\hol$, eight-vertex model and complex $\ceo$. The complexity of the first four problems has been fully classified as four FP vs. \#P dichotomies by \cite{cai2013vanishing,backens2017holantc,shao2020realholant,caifu2023eightvertex} respectively. In each dichotomy, the corresponding problem is either polynomial-time computable or \#P-hard, given a set of signatures that defines the problem. On the other hand, despite the advancements made in the study of the complexity of complex \ceo\ \cite{cai2018complexity,cai2020beyond,shao2024eulerian,meng2024p}, a complete classification for complex \ceo\ remains unsolved.

Complex \ceo, which is an abbreviation of counting complex-weighted Eulerian orientations problems, asks for the sum of the weights of all Eulerian orientations in a given graph. This problem is derived from the partition function of the ice-type model in statistical physics \cite{pauling1935structure}, and was generalized into a mathematical problem in \cite{mihail1996number}. Furthermore, this problem has been extended to a weighted version in \cite{cai2018complexity} under the framework of \hol. Please refer to Definition \ref{def:numbereo} for a detailed version of this definition.

In this article, we always restrict counting problems to Boolean domain and complex range by default. We present an $\pnp$ vs. \#P dichotomy for complex \ceo. In this dichotomy, we demonstrate that for each signature set $\mathcal{F}$ that defines \ceo, $\ceo(\mathcal{F})$ is either \#P-hard or can be computed in polynomial time with a specific NP oracle. We remark that the NP oracle we use here might be replaced with a polynomial-time algorithm, hence this dichotomy has the potential to achieve an FP vs. \#P dichotomy.  This dichotomy generalizes the former results presented in  \cite{cai2018complexity,cai2020beyond,shao2024eulerian,meng2024p}, with the exception of some polynomial-time algorithms developed therein. Furthermore, we illustrate that three $\pnp$ vs. \#P dichotomies for \hol\ problems defined on \upside, \downside\ and \sw\ signatures, respectively, can also be obtained from this dichotomy.

\subsection{Introduction of the counting problems}\label{subsec:intro}

 In this section, we present an informal version of the definitions of \hol,  \ceo\ and \#\csp. See Definition \ref{defHol}, \ref{def:numbereo} and \ref{def:csp} for their detailed definitions. A signature $f:\{0,1\}^k\to \mathbb{C}$ is defined as an \eo\ signature if $k$ is even and $f(\alpha)=0$ whenever the number of $0'$s does not equal to that of $1'$s in the string $\alpha$. An orientation variable is a Boolean variable that assigns one of two possible orientations to an undirected edge. In this context, the orientation is represented by the value $0$ assigned to the tail and the value $1$ assigned to the head of the edge.

 \begin{definition}[\hol,  \ceo\ and\#\csp]
      \hol\ is defined by a signature set $\mathcal{F}$, denoted as $\hol(\mathcal{F})$. An instance (input) of $\hol(\mathcal{F})$ is a graph $G=(V,E)$, with each $v\in V$ assigned a signature from $\mathcal{F}$ and each $e\in E$ represents a Boolean variable. The output is the sum of the product of the values of all the signatures, calculated over all possible assignments to the variables.

      The definition of $\ceo$ is identical to that of $\hol$, with the exception that each $v\in V$ is assigned an $\eo$ signature and each $e\in E$ represents an orientation variable.
      
       The definition of $\#\csp$ is identical to that of $\hol$, with the exception that $G$ is a hypergraph, which means that each $e\in E$ may be incident to $k\ge 1$ vertices in $V$ instead of exactly two of them. 
       \label{definformal}
  \end{definition}

  \hol\ and \#\csp\ play a pivotal role in the field of counting complexity, due to their capacity to express a diverse range of natural counting problems. For example, counting the number of independent sets in a graph as well as counting the number of solutions in a 3-CNF formula can be expressed as a \#\csp\ problem, while counting the number of perfect matchings and in a graph, \#\csp\ as well as \ceo\ can be transformed into $\hol$ problems. Though the complexity classification for \#\csp\ has been presented as a dichotomy by \cite{cai2014complexity}, the complexity of \hol\ remains uncertain. In particular, \ceo\ captures a specific category of problems related to vanishing signatures in \hol. By contrast, its complexity remains undetermined. This leads us to the necessity of studying the complexity of \ceo.

\subsection{Significance of \ceo}

As illustrated in Section \ref{subsec:intro}, the complexity classification of \ceo\ has the potential to lead to a complete dichotomy for \hol. Apart from that, \ceo\  also plays a pivotal role in the field of statistical physics and
combinatorics. In statistical physics, the output of \ceo\ is denoted as the partition function of the ice-type model when the underlying graph $G$ is restricted to some specific graph classes \cite{pauling1935structure,slater1941theory,rys1963uber}, and  in particular as the partition function of the classical six-vertex model when $G$ is restricted to a finite region of the square lattice \cite{pauling1935structure}. An exact solution of the classical six-vertex model with periodic boundary conditions was demonstrated by Lieb, which is a considerable milestone in statistical physics\cite{lieb1967residual}. 
In combinatorics, the resolution of the alternating sign matrix conjecture is also related to the classical six-vertex model \cite{korepin1982calculation,mills1983alternating,zeilberger1994proof,kuperberg1996another,bressoud1999proofs}. Furthermore, the stochastic six-vertex model is defined by a specific signature grid of \ceo\ in the quadrant $\mathbb{Z}_{\ge 1}\times \mathbb{Z}_{\ge 1}$ \cite{borodin2016stochastic}, whose partition function has an  intriguing connection with the Hall-Littlewood processes \cite{borodin2016between} and q-Whittaker processes \cite{orr2017stochastic}. 
\subsection{Our results}

Our main result can be stated as follows.
\begin{theorem}
Let $\mathcal{F}$ be a set of EO signatures. Then $\ceo(\mathcal{F})$ is either in $\pnp$ or \#P-hard.
    \label{thm:dicoceo}
\end{theorem}
A detailed version of Theorem \ref{thm:dicoceo}, which presents the explicit criterion for this dichotomy, is stated as Theorem \ref{thmdetail} in Section \ref{secdetail}. It is noteworthy that two distinct types of extension have been applied to Theorem \ref{thm:dicoceo}, resulting in three additional dichotomies for \upside, \downside\ and \sw\ signatures, respectively, as presented in Corollary \ref{coro:upside} and \ref{coro:sw}. It is our contention that this dichotomy has the potential to classify the complexity of more cases within the framework of \hol, due to the fact that a dichotomy for real \hol\ \cite{shao2020realholant} is based on the research of $\ceo$ problems defined by signatures with ARS property \cite{cai2020beyond}, and a dichotomy for eight-vertex model \cite{caifu2023eightvertex} is based on the research of $\ceo$ problem known as six-vertex model \cite{cai2018complexity}.

We also remark that an $\pnp$ vs. \#P dichotomy does separate the complexity of $\ceo$. The definitions of complexity classes in the following can be found in \cite{arora2009computational}. If $\text{\#P}\subseteq \pnp$, then we have 

$$\text{PH}\subseteq \text{FP}^\text{\#P}\subseteq \text{FP}^{\pnp}=\pnp=\Delta_2^\text{P},$$
which indicates that the complexity class PH collapses at the second level. Conversely, if we assume PH does not collapse, it follows that a  \#P-hard problem cannot be solved by an $\pnp$ algorithm.

Furthermore, the complexity classification of the problem that defines the NP oracle leads us to the study of Boolean constraint satisfaction problems \cite{feder2006classification}. Nevertheless, it is regrettable that the complexity of the problem defining the NP oracle has not been addressed by existing research.

\subsection{Our approaches}
In the proof of the dichotomy for \ceo,  we notice that the pinning signature $\Delta$, which has the ability to fix a variable to 0 and another to 1 at the same time, is of great significance. Based on this observation, our proof basically consists of 2 parts, which are realizing $\Delta$ and proving a dichotomy with the existence of $\Delta$.

In the first part of the proof, we develop a generating process to realize $\Delta$. Using this process, we may realize $\Delta$, mainly through gadget constructions and polynomial interpolation. In particular, sometimes we also realize $\Delta$ by duplicating a single disposable $\Delta$ via gadget constructions. Nevertheless, there are circumstances in which it is not possible to realize $\Delta$. However, it is fortunate that in such cases, either the signature set can induce the \#P-hardness, or is composed of signatures with the ARS property ignoring a constant, whose complexity has already been classified by \cite{cai2018complexity}. Consequently, we may assume the existence of $\Delta$ in the subsequent proofs.

In the second part of the proof, we find out that the bitwise addition of three strings belonging to the support of the signature is the most appropriate criterion for the categorization of situations. For the hardness part of the proof, this criterion can capture the certificate indicating that the support is not affine, and thereby a reduction similar to that used in the dichotomy for \#CSP can be applied. For the algorithm part of the proof, this criterion establishes the closure property to certain kind of signatures, which ultimately gives rise to an $\pnp$ algorithm. Furthermore, the dichotomy for pure signatures  demonstrated in \cite{meng2024p} can also be applied under this criterion.

Furthermore, by the fact that an orientation assigns equal number of 0's and 1's to the variables, a dichotomy for \upside\  and \downside\ signatures, respectively, can be obtained from the dichotomy for \ceo. Moreover, the addition of a suitable number of pinning signatures can transform each \sw\ signature into an \eo\ signature, and this process can consequently lead to a dichotomy for \sw\ signatures.

\subsection{Organization}
In Section \ref{sec:preli}, we introduce preliminaries needed in our proof. 
In Section \ref{secdetail}, we present a detailed version of the dichotomy for \ceo. 
In Section \ref{sec:ceo}, we reduce $\ceoc$ to \ceo. 
    In Section \ref{section: dichotomy for eoc}, we demonstrate an $\pnp$ vs. \#P dichotomy for $\ceoc$. 
In Section \ref{sec:extend}, we extend the dichotomy for \ceo\ to other signature sets.
In Section \ref{sec:conclusion}, we conclude our results.

\section{Preliminary}\label{sec:preli}

\subsection{Definitions and notations}
A \textit{Boolean variable} is defined over a Boolean domain, commonly represented by the symbols $\{0,1\}$. This is sufficient for defining Boolean counting problems. However, in certain contexts, it is necessary to regard Boolean domain as the finite field $\mathbb{F}_2$ to identify particular classes of tractable signatures. 

A \textit{signature} $f$ with $r$ variables is a mapping from $\{0,1\}^r$ to $\mathbb{C}$, where $f(\alpha)$ or $f_\alpha$ denotes the value of $f$ on input string $\alpha$. We briefly call $\alpha$ a string for convenience. The set of variables of $f$ is denoted by $\text{Var}(f)$, and its size is denoted by $\text{arity}(f)$. The \textit{support} of $f$, denoted by $\su(f)$, is the set of all strings on which $f$ takes non-zero values.

For any binary string $\alpha$ and $s \in \{0,1\}$, $\#_s(\alpha)$ represents the number of $s$'s in $\alpha$. $\#_1(\alpha)$ is widely known as the \textit{Hamming weight} of $\alpha$. We use $\len\alpha$ to denote the length and $\alpha_i$ to denote the $i$th bit of $\alpha$. We use $\overline\alpha$ to denote the \textit{dual string }of $\alpha$, obtained by flipping every bit of $\alpha$. We use $\dis{\alpha}{\beta}$ to denote the \textit{distance} between $\alpha$ and $\beta$, which means the number of different bits of the two strings. There are specific sets defined as the following:

$$
\eoe = \{\alpha \mid \#_1(\alpha) = \#_0(\alpha)\} \quad \text{and} \quad \eog = \{\alpha \mid \#_1(\alpha) \geq \#_0(\alpha)\}.
$$

A signature $f$ is referred to as a $\eoe$ signature, or simply an $\eo$ signature, if $\su(f) \subseteq \eoe$. By definition, this implies that $\text{arity}(f)$ is even. The term "$\eo$ signature" was introduced in \cite{cai2020beyond} as an abbreviation for \textit{Eulerian Orientation signature}. If $\su(f) \subseteq \eog$, we refer to $f$ as an $\eog$ signature. Similar notations are also defined, which replace "$\geq$" with "$\leq$", "<", or ">" in defining formulas and names.

There is an important operation between strings called \textit{bitwise addition}, usually denoted by $\oplus$. For two strings $\alpha$ and $\beta$, the bitwise addition $\alpha\oplus\beta$ is a string $\gamma$, such that $\gamma_i=(\alpha_i+\beta_i)\text{ mod }2$. Suppose $A$ is a set of strings. We call $A$ \textit{affine} when for any three strings $\alpha,\beta,\gamma\in A$ (no need to be distinct), $\alpha\oplus\beta\oplus\gamma\in A$. For a set $S$ of 01-strings of length $k$, the \textit{affine span} of $S$ is defined as the minimal affine space that contains $S$. For a signature $f$, we denote the affine span of its support by $\text{Span}(f)$. 

A signature $f$ is \textit{symmetric} when its values are merely related to the Hamming weight of inputs. In this context, a symmetric signature $f$ of arity $r$ is represented as $[f_0,f_1,\ldots,f_r]$, where $ f_i $ is the evaluation of $ f $ on all input strings with Hamming weight $ i $. By swapping the symbols in the Boolean domain, we obtain the \textit{dual} signature of $ f $, denoted by $\widetilde{f} = [\widetilde{f}_0, \widetilde{f}_1, \ldots, \widetilde{f}_r] = [f_r, f_{r-1}, \ldots, f_0]$. A signature $ f $ is called \textit{self-dual }if $ f = \widetilde{f} $. A special unary signature denoted by $\Delta_0$ is $[1,0]$, with its dual represented by $\Delta_1=[0,1]$. The binary disequality signature, denoted by $\neq_2$, is represented as $[0,1,0]$, which is also self-dual. The notions of dual and self-dual can be naturally extended to asymmetric signatures, sets of signatures, Boolean domain counting problems, and tractable classes within dichotomy theorems. 

We denote \textit{tensor multiplication} of signatures by $\otimes$.

Several classes of frequently used signatures are defined as follows. We use $\mathcal{EQ}$ to represent the set of \textit{equality signatures}, where $\mathcal{EQ} = \{=_1, =_2, \dots, =_r, \dots\}$. Here, $=_r$ represents an arity $ r $ signature $[1, 0, \dots, 0, 1]$, taking the value 1 when the input is either all 0's or all 1's, and the value 0 otherwise. A \textit{disequality signature} of arity $ 2d $, denoted by $\neq_{2d}$, is an asymmetric signature. Under the condition that its variables are specifically ordered, it takes value 1 when $ x_1 = x_2 = \ldots = x_d \neq x_{d+1} = x_{d+2} = \ldots = x_{2d} $ and 0 otherwise. We use $\mathcal{DEQ} = \{\neq_2, \neq_4, \dots, \neq_{2n}, \dots\}$ to denote the set of all disequality signatures. It is worth noticing that $\mathcal{DEQ}$ is closed under variable permutation. For any permutation $\pi \in S_{2d}$, a signature $ f $ that takes value 1 when $ x_{\pi(1)} = x_{\pi(2)} = \ldots = x_{\pi(d)} \neq x_{\pi(d+1)} = x_{\pi(d+2)} = \ldots = x_{\pi(2d)} $ (and 0 otherwise) is also regarded as a disequality signature.

An alternative definition of disequality signatures applies to $\eo$ signatures: if $ f $ is an $\eo$ signature of arity $ 2d $ with a string $ \alpha \in \eoe $ such that $\su(f) = \{\alpha, \overline{\alpha}\}$ and $ f(\alpha) = f(\overline{\alpha}) = 1 $, then $ f $ is a disequality signature. It is straightforward to verify that this alternative definition is equivalent to the original one.

$f$ is called a \textit{generalized disequality signature}, denoted by $\neq_{2d}^{a,b}$, when it is an $\eo$ signature of arity $ 2d $ such that there exists a string $\alpha \in \eoe$ satisfying $\su(f) = \{\alpha, \overline{\alpha}\}$, where $ f(\alpha) = a $ and $ f(\overline{\alpha}) = b $.

Now we provides a formal definition of counting weighted Eulerian orientations problems ($\#\eo$ problems), $\hol$ problems, and counting constraint satisfaction problems ($\#\csp$ problems) and outlines their relationships. We largely refer to \cite[Section 1.2]{cai2017complexity} and \cite[Section 2.1]{cai2020beyond}.

Let $\mathcal{F}$ denote a fixed finite set of $\eo$ signatures. An $\eo$-signature grid, $\Omega(G, \pi)$, over $\mathcal{F}$ is constructed from an Eulerian graph $ G = (V, E) $ without isolated vertices, which means that every vertex has a positive even degree. Here, $\pi$ assigns each vertex $ v \in V $ a signature $ f_v \in \mathcal{F} $, along with a fixed linear order of its incident edges. The arity of each $ f_v $ matches the degree of $ v $, with each incident edge corresponding to a variable of $ f_v $. For any Eulerian graph $ G $, let $ \eo(G) $ denote the set of all possible Eulerian orientations of $ G $. Given an orientation in $ \eo(G) $, each edge is assigned a direction. For each edge with two ends, we assign $ 0 $ to the head and $ 1 $ to the tail, so each Eulerian orientation corresponds to an assignment of Boolean values to the edges' ends, ensuring an equal number of $ 0 $’s and $ 1 $’s incident to each vertex. Under this assignment, a vertex $ v $ contributes a weight, based on $ f_v $ and the local assignment to its incident edges. Each $ \sigma \in \eo(G) $ provides an evaluation $ \prod_{v \in V} f_v(\sigma|_{E(v)}) $, where $ \sigma|_{E(v)} $ is the restriction of $ \sigma $ to the edges' ends incident to $ v $. More specifically speaking, $\sigma|_{E(v)}$ assigns 0 to an incoming edge and 1 to an outgoing edge.

\begin{definition}[$\#\eo$ problems]\cite{cai2020beyond}\label{def:numbereo}
A $\#\eo$ problem $\#\eo(\mathcal{F})$, parameterized by a set $\mathcal{F}$ of $\eo$ signatures, is defined as follows: Given an instance $I$, or equivalently an $\eo$-signature grid $\Omega(G, \pi)$ over $\mathcal{F}$, the output is the partition function of $\Omega$, 

$$
\text{Z}(I)=\#\eo_\Omega = \sum_{\sigma \in \eo(G)} \prod_{v \in V} f_v(\sigma|_{E(v)}).
$$
\end{definition}
Six-vertex model can be seen as a special case of $\#\eo$ problems defined by a single quaternary \eo\ signature. Moreover, $\#\eo$ problems constitute a subclass of $\hol$ problems. A (general) signature grid over an arbitrary set $ \mathcal{F} $ of signatures (not limited to $\eo$ signatures) follows a similar definition. Here, $ G $ is a general graph, and $ \pi $ assigns a signature to each vertex, along with an order on its incident edges. We consider all possible $ \{0,1\} $-assignments $ \sigma : E(G) \rightarrow \{0,1\}^{|E|} $, and each assignment $ \sigma $ gives an evaluation $ \prod_{v \in V} f_v(\sigma|_{E(v)}) $.

\begin{definition}[$\hol$ problems]
    A $\hol$ problem, $\hol(\mathcal{F})$, parameterized by a set $ \mathcal{F} $ of complex-valued signatures, is defined as follows: Given an instance $I$, or equivalently a signature grid $ \Omega(G, \pi) $ over $ \mathcal{F} $, the output is the partition function of $ \Omega $,
    
    $$
    \text{Z}(I)=\hol_\Omega = \sum_{\sigma:E\rightarrow\{0,1\}} \prod_{v \in V} f_v(\sigma|_{E(v)}).
    $$
The bipartite $\hol$ problem $\hol(\mathcal{F} \mid \mathcal{G})$ is a specific case where $ G $ is bipartite, say $ G(U, V, E) $, with vertices in $ U $ assigned signatures from $ \mathcal{F} $ and vertices in $ V $ assigned signatures from $ \mathcal{G} $. We denote $U$ as the left-hand side, or LHS for short. We denote $V$ as the right-hand side, or RHS for short.
\label{defHol}
\end{definition}

In this article, the \textit{size }of an instance $I=\Omega(G,\pi)$ is defined as $|G|$. We use $ \leq_T $ and $ \equiv_T $ to denote polynomial-time Turing reductions and equivalences, respectively. Notably, it is known that $\#\eo$ problems can be represented as a special case of bipartite $\hol$ problems, which can be obtained by operating a holographic transformation by $Z=\frac{1}{\sqrt{2}}\left(\begin{matrix}
1 &  1\\
\mathfrak{i} & -\mathfrak{i} 
\end{matrix}\right)$ on $\hol$ problems.

\begin{lemma}\cite{cai2020beyond}
    $\#\eo(\mathcal{F}) \equiv_T \hol(\neq_2 \mid \mathcal{F})$.
\end{lemma}

\begin{lemma}\cite{shao2020realholant}
    $\hol(\neq_2\mid\mathcal{F})\equiv_T \hol(Z^{-1}\mathcal{F}).$
\end{lemma}

Finally, we introduce counting constraint satisfaction problems ($\#\csp$). 

\begin{definition}\label{def:csp}
A $\#\csp$ problem is defined by a finite set $ \mathcal{F} $ of signatures, where each signature has complex values and is defined on the Boolean domain. Given an instance of $ \#\csp(\mathcal{F}) $ consisting of a finite set of variables $ \{x_1, x_2, \ldots, x_n\} $ and a set of clauses $ C $. Each clause in $ C $ includes a signature $ f $ from $ \mathcal{F} $ and a selection of variables from $ \{x_1, x_2, \ldots, x_n\} $ (repetitions allowed). Suppose the ith clause $C_i$ includes a signature $ f_i=f $ of arity $ k $, then $ k $ relevant variables are $x_{i_1},x_{i_2},\ldots,x_{i_k}$. The output is given by

$$
\sum_{x_1, x_2, \ldots, x_n \in \{0,1\}} \prod_{(f, x_{i_1}, x_{i_2}, \ldots, x_{i_k}) \in C} f(x_{i_1}, x_{i_2}, \ldots, x_{i_k}).
$$
\end{definition}
It is well-known that $\#\csp$ problems can also be expressed as $\hol$ problems. Specifically, from \cite[Lemma 1.2]{cai2017complexity}, we have:

\begin{lemma}\cite{cai2017complexity}
$\#\csp(\mathcal{F}) \equiv_T \hol(\mathcal{EQ} \cup \mathcal{F}).$
\end{lemma}

\subsection{Gadget construction and signature matrix}

This subsection introduces two fundamental concepts: \textit{gadget construction} and the signature matrix, both of which play critical roles in complexity classification through reduction techniques. Gadget construction is essential for constructing reductions, while the signature matrix provides an intuitive representation of a signature, linking gadget construction to matrix multiplication.

Let $\mathcal{F}$ be a set of signatures. An $\mathcal{F}$-gate resembles a signature grid $\Omega(G,\pi)$ but differs in that $G = (V, E, D)$ includes two types of edges: internal edges (connecting 2 vertices in $V$ and contained in $E$) and dangling edges (connecting to only 1 vertex in $V$ and contained in $D$). An $\mathcal{F}$-gate represents a signature, with variables corresponding to the edges in $D$. Suppose $|E| = n$ and $|D| = m$, and let the edges in $E$ represent variables $\{x_1, x_2, \ldots, x_n\}$ while edges in $D$ represent variables $\{y_1, y_2, \ldots, y_m\}$. The $\mathcal{F}$-gate then defines a signature $f$ as:

$$
f(y_1, y_2, \ldots, y_m) = \sum_{\sigma: E \rightarrow \{0,1\}} \prod_{v \in V} f_v(\hat{\sigma}|_{E(v)}),
$$
where $(y_1, y_2, \ldots, y_n) \in \{0,1\}^n$ denotes the assignment to the variables represented by the dangling edges, $\hat{\sigma}$ extends $\sigma$ by incorporating this assignment, and $\pi$ assigns $f_v$ to each vertex $v \in V$. The signature $f$ obtained by an $\mathcal{F}$-gate is referred to as being realizable from $\mathcal{F}$ through gadget construction. When no internal edges exist in an $\mathcal{F}$-gate, the resulting signature becomes the tensor product of several signatures. The operation called tensor multiplication is denoted by $\otimes$.

As shown in \cite[Lemma 1.3]{cai2017complexity}, if a signature $f$ is realizable by $\mathcal{F}$, it follows that:

$$
\hol(\mathcal{F}) \equiv_T \hol(\mathcal{F} \cup \{f\}),
$$
which is a powerful tool toward reductions. In addition, the closure property under gadget construction is widely used. As shown in \cite[Chapter 3]{cai2017complexity}, the famous classes $\mathscr{A}$ and $\mathscr{P}$ are closed under gadget construction, implying that if $\mathcal{F} \subseteq \mathscr{A}$ or $\mathscr{P}$, then all $\mathcal{F}$-gates belong to $\mathscr{A}$ or $\mathscr{P}$, respectively.

In the setting of $\#\eo$ problems, the concept of $\mathcal{F}$-gate is slightly different since $\#\eo(\mathcal{F})$ can be viewed as a bipartite $\hol$ problems, $\hol(\neq_2 \mid \mathcal{F})$. For each internal edge, an additional vertex assigned $\neq_2$ is added, and we view this modified signature grid as a general $\hol$ gate. In this paper, $\mathcal{F}$-gate is defined this way by default in the analysis of $\#\eo$ problems.
As an application of gadget construction, in this paper, when we construct $\neq_4^{a,b}$ with $ab\neq0$, we can further construct $\ouhe$ by connecting two edges between two copies of $\neq_4^{a,b}$ properly. We formalize this fact as a lemma for future reference.

\begin{lemma}
    Suppose $\neq_4^{a,b}\in \mathcal{F}, ab\neq 0$. Then $\ceo(\mathcal{F}\cup\{\ouhe\})\le_T\ceo(\mathcal{F})$.
    \label{lem:neq4ab obtain ouhe}
\end{lemma}
\begin{proof}
    Without loss of generality let $\su(\neq_4^{a,b})=\{0101,1010\}$. By connecting the third variables of 2 copies of $\neq_4^{a,b}$ and then the fourth variables of these 2 copies via $\neq_2$, we realize the $\ouhe$ signature.
\end{proof}
Next we introduce the matrix form of signatures. The signature matrix serves as an organized expression of all possible values for a given signature, providing a convenient way to calculate gates. Suppose $f$ is a signature mapping $\{0,1\}^r$ to $\mathbb{C}$. First, assign an order to all variables of $f$, say $x_1, x_2, \ldots, x_r$. The signature matrix with parameter $l$ is a $2^l \times 2^{r - l}$ matrix for some integer $0 \leq l \leq r$, where the values of $x_1, x_2, \ldots, x_l$ form the row indices and those of $x_{l+1}, x_{l+2}, \ldots, x_r$ form the column indices. The entries are corresponding values of $f$ given the input string that combines the row index and the column index. For even arity signatures, we often take $l = \frac{r}{2}$. The matrix is denoted as $M(f)_{x_1x_2\cdots x_l,x_{l+1}x_{l+2}\cdots x_r}$, sometimes for convenience we call it $M_f$ without causing ambiguity.

We present two examples of signature matrix:

1. Binary signature: For a binary signature $f = (f_{00}, f_{01}, f_{10}, f_{11})$, the matrix form is:

$$
M_f = M_{x_1, x_2} = \begin{pmatrix} f_{00} & f_{01} \\ f_{10} & f_{11} \end{pmatrix}.
$$

2. Quaternary signature: For an arity-4 signature $g$, we have:

$$
M_g = M_{x_1x_2, x_3x_4} = \begin{pmatrix} g_{0000} & g_{0001} & g_{0010} & g_{0011} \\ g_{0100} & g_{0101} & g_{0110} & g_{0111} \\ g_{1000} & g_{1001} & g_{1010} & g_{1011} \\ g_{1100} & g_{1101} & g_{1110} & g_{1111} \end{pmatrix}.
$$

We then introduce the relation Between matrix multiplication and gadget construction.

Consider two $\mathcal{F}$-gates $f$ and $g$ given some set $\mathcal{F}$, of arity $n$ and $m$, respectively. Let $\text{Var}(f) = \{x_1, x_2, \ldots, x_n\}$ and $\text{Var}(g) = \{y_1, y_2, \ldots, y_m\}$. We construct a new gadget by connecting part of their dangling edges corresponding to subsets $\{x_{n - l + 1}, \ldots, x_n\}$ and $\{y_1, y_2, \ldots, y_l\}$ of $\text{Var}(f)$ and $\text{Var}(g)$, respectively, where $l \in \mathbb{Z}_+$. The resulting signature $h$ is also an $\mathcal{F}$-gate with $\text{Var}(h) = \{x_1, \ldots, x_{n - l}, y_{l + 1}, \ldots, y_m\}$, and by the definition of gadget construction and matrix multiplication, it is easily verified that:

$$
M_h = M_{x_1 \ldots x_{n - l}, y_{l + 1} \ldots y_m} = M_{x_1 \ldots x_{n - l}, x_{n - l + 1} \ldots x_n} \cdot M_{y_1 \ldots y_l, y_{l + 1} \ldots y_m} = M_f \cdot M_g.
$$
In the setting of $\#\eo$ problems, each edge is viewed as $\neq_2$ rather than $=_2$, leading to the slightly different form:

$$
M_h = M_f \cdot \neq_2^{\otimes l} \cdot M_g.
$$
This notation is frequently used in subsequent discussions without further explanation.

Two particular constructions are especially notable:

1. Adding self-loop: In $\#\eo$ problems, adding a self-loop to a signature $f$ means selecting two variables, say $x_1$ and $x_2$, and adding an internal edge between them to form a new signature $f'$. If $\text{Var}(f) = \{x_1, x_2, \ldots, x_n\}$, the new signature $f'$ is:

$$
   f'(x_3, \ldots, x_n) = f(0, 1, x_3, \ldots, x_n) + f(1, 0, x_3, \ldots, x_n),
$$
often denoted by $f^{x_1 \neq x_2}$.

One can directly verify that gadget construction can be seen as a combination of tensor multiplication and adding self-loops. In addition, the result of adding a self-loop to a binary signature is a complex number, which is also the partition function of this simple signature grid.

We then generalize this operation by a binary signature denoted by $\bi{a}{b}$, which has the matrix form 
$
\begin{pmatrix}
0 & a \\
b & 0
\end{pmatrix}
,$ and is referred to as a generalized binary disequality. The process, called adding a self-loop by $\bi{a}{b}$ on $f$, resembles previous operations but differs in that the internal edge is not directly added between $x_1$ and $x_2$. Instead, an additional vertex assigned $\bi{a}{b}$ is added, and we connect two dangling edges incident to $x_1$ and $x_2$ respectively to the dangling edges of $\bi{a}{b}$. This construction can be viewed as a weighted self-loop. Notably, there are two possible situations, that is either

$$
f'(x_3, \ldots, x_n) = af(0,1,x_3, \ldots, x_n) + bf(1,0,x_3, \ldots, x_n),
$$ 
or 

$$
f'(x_3, \ldots, x_n) = bf(0,1,x_3, \ldots, x_n) + af(1,0,x_3, \ldots, x_n),
$$
depending on the direction of $\bi{a}{b}$.

2. Pinning: The pinning operation fixes a variable to a specific value. In the setting of $\hol$ or $\#\csp$ problems, pinning a variable to 0 or 1 is represented by connect $\Delta_0 = (1, 0)$ or $\Delta_1 = (0, 1)$ to its corresponding edge. In $\#\eo$ problems, pinning means adding a self-loop by $\bi{1}{0}$. For variables $x$ and $y$ in $\text{Var}(f)$, by the pinning operation we can fix $x$ to 1 and $y$ to 0 at the same time, maintaining the $\eo$ property. The resulting signature is denoted by $f^{xy=10}$. We can also fix more than $2$ variables of a signature $f$, say $(x_1x_2\ldots x_t)$ is fixed to $\sigma$, where $\sigma$ is a 01-string with length $t$, then the resulting signature is denoted by $f^{x_1x_2\ldots x_t=\sigma}$, or briefly denoted by $f^{\sigma}$ without causing ambiguity. In this paper, when referring to the pinning signature $\pin$, we mean $\bi{1}{0}$.

\subsection{Previous results}

In this section we present several previous results, which are foundations of our work. 

First we present the dichotomy theorem for $\#\csp$ problems. To describe the tractability conditions of signatures in $\#\csp$ problems, we introduce two significant tractable classes, $\mathscr{A}$ and $\mathscr{P}$.

Let $ X $ denote a $(d+1)$-dimensional column vector over $ \mathbb{F}_2 $, such that $ X = (x_1, x_2, \ldots, x_d, 1)^T $, and suppose $ A $ is a matrix over $ \mathbb{F}_2 $. The function $ \chi_{AX} $ takes value 1 when $ AX = \mathbf{0} $ and value 0 otherwise. This function describes an affine relation on the variables $ x_1, x_2, \ldots, x_d $.

\begin{definition}
We denote by $\mathscr{A}$ the set of signatures which have the form $\lambda\chi_{AX}\cdot\mathfrak{i}^{L_1(X)+L_2(X)+\cdots+L_n(X)}$, where $\mathfrak{i}=\sqrt{-1}$, $\lambda\in\mathbb{C}$, $n\in\mathbb{Z}_+$, each $L_j$ is a 0-1 indicator function of the form $\langle \alpha_j,X \rangle$, where $\alpha_j$ is a $(d+1)$-dimensional vector over $\mathbb{F}_2$, and the dot product $\langle \cdot,\cdot \rangle$ is computed over $\mathbb{F}_2$.
\label{defa}
\end{definition}

\begin{definition}
 We denote by $\mathscr{P}$ the set of all signatures which can be expressed a product of unary signatures, binary equality signatures ($=_2$) and binary disequality signatures ($\neq_2$) (on not necessarily disjoint subsets of variables). 
 \label{defp}
\end{definition}

Now we can state a pivotal theorem that classifies Boolean $\#\csp$ problems as follows:

\begin{theorem}[{\cite[Theorem 3.1]{cai2014complexity}}]\label{thm:CSPdichotomy}
Suppose $\mathcal{F}$ is a finite set of signatures mapping Boolean inputs to complex numbers. If $\mathcal{F}\subseteq\mathscr{A}$ or $\mathcal{F}\subseteq\mathscr{P}$, then $\#\csp(\mathcal{F})$ is computable in polynomial time. Otherwise, $\#\csp(\mathcal{F})$ is \#P-hard.
\end{theorem}

Both classes $ \mathscr{A} $ and $ \mathscr{P} $ are self-dual and defined by signatures with affine supports. The requirement for affine supports is a necessary tractability condition here, with further restrictions on signature weights (for $ \mathscr{A} $) or more specialized form of support (for $ \mathscr{P} $).

This pattern, which implies that the affine support may somehow contribute to tractability, will serve as a critical idea in the analysis of other dichotomies. Specifically, for the six-vertex model, this necessary structural conditions extend the idea of affine support in novel ways, as is observed in \cite{cai2018complexity}. In the $\#\eo$ framework, we can generalize this observation to a more universal form as follows, while present the relationship between $\#\eo$ and $\#\csp$ problems.

Let $ f $ be an $\eo$ signature of arity $ 2d $ and over a variable set $ \text{Var}(f) = \{x_1, x_2, \ldots, x_{2d}\} $. For an arbitrary perfect pairing $ P $ of $ \text{Var}(f) $, say $ P = \{\{x_{i_1}, x_{i_2}\}, \{x_{i_3}, x_{i_4}\}, \ldots, \{x_{i_{2d-1}}, x_{i_{2d}}\}\} $, we define $ \eom[P] $ as the subset of $ \{0,1\}^{2d} $ satisfying the following condition:

$$
\eom[P] = \{\alpha \in \{0,1\}^{2d} \mid \alpha_{i_1} \neq \alpha_{i_2}, \ldots, \alpha_{i_{2d-1}} \neq \alpha_{i_{2d}}\}.
$$

When $ \su(f) \subseteq \eom[P] $, $ f $ is referred to as an $ \eom[P] $ signature, corresponding to a fixed $P$. If there exists some perfect pairing $ P $ for which $ f $ meets this requirement, we refer to $ f $ as a pairwise opposite (defined specifically in \cite{cai2020beyond}) or $ \eom $ signature.

For a given $ \eom[P] $ signature $ f $, we define $\tau_f$ (or $\tau(f)$) as the set consisting of arity-$ d $ signatures under the disequality constraints: 

$$
\tau_f = \{f(x_{j_1}, 1 - x_{j_1}, x_{j_2}, 1 - x_{j_2}, \ldots, x_{j_{d}}, 1 - x_{j_{d}}) \mid x_{j_k} \in \{x_{i_{2k-1}}, x_{i_{2k}}\},k\in\mathbb{N}_+\}.
$$
Notably, variables in one pair can be exchanged freely without changing $P$, therefore $|\tau_f|=2^d$.

An inverse mapping $ \pi $ can also be defined. Suppose $g$ is an arbitrary signature of arity $ d $, $\pi_g$ (or $\pi(g)$) is defined as:

$$
\pi_g(x_1, y_1, x_2, y_2, \ldots, x_d, y_d) = g(x_1, x_2, \ldots, x_d) \cdot \neq_2(x_1,y_1) \cdots \neq_2(x_d,y_d).
$$
It can be easily verified that $ \pi_g $ is an $ \eom[P] $ signature corresponding to the perfect pairing $ P = \{\{x_1, y_1\}, \{x_2, y_2\}, \ldots, \{x_d, y_d\}\} $, in particular, $\tau(\pi(g))$ contains $g$. $\pi$ and $\tau$ can both maintain the property of having an affine support and the additional condition corresponding to $\mathscr{A}$ and $\mathscr{P}$. The following lemma can be directly verified by definition. 

\begin{lemma}\label{lem:tau  pi maintain A P}
For an $\eom$ signature (even an $\eo$ signature) $f$, $f\in\mathscr{A}$ (or $\mathscr{P}$) is equivalent to that $\tau(f)\subseteq  \mathscr{A}$ (or $\mathscr{P}$).
For a signature $g$, $g\in\mathscr{A}$ (or $\mathscr{P}$) is equivalent to that $\pi_g\in\mathscr{A}$ (or $\mathscr{P}$).
\end{lemma}

Furthermore, $\tau(\pi(g))$ contains all signatures attainable from flipping part of $g$'s variables using $\neq_2$. This mapping ensures computational equivalence between $\#\eo$ and $\#\csp$ with a free binary disequality.

\begin{theorem}[{\cite[Theorem 6.1]{cai2020beyond}}]\label{thm:csp=eom}
For any signature set $ \mathcal{F} $,

$$
\#\csp(\mathcal{F}) \equiv_T \#\eo(\pi(\mathcal{F})).
$$
\end{theorem}

This equivalence along with Theorem \ref{thm:CSPdichotomy} implies a dichotomy for $\#\eo$ problems defined by $ \eom $ signatures. The dichotomy with two tractable classes based on $ \mathscr{A} $ and $ \mathscr{P} $ types can be stated as follows: if $ f $ is an $ \eom $ signature, and it satisfies the condition $ \tau(f) \subseteq \mathscr{A} $ or $ \tau(f) \subseteq \mathscr{P} $, then $\#\eo(f)$ is polynomial-time computable, otherwise it is \#P-hard. Noticing that two tractable classes share the property that $f$ is $\eom$ and each signature in $\tau(f)$ is affine. In fact, the former one is equivalent to that $f$ is $\eo$ and $\su(f)$ is affine.

\begin{lemma}[{\cite[Lemma 5.7]{cai2020beyond} }] \label{lem: eo+affine= eom}
If an $\eo$ signature $f$ has affine support, then $f$ is a pairwise opposite signature ($\eom$ signature). 
\end{lemma}

Now we introduce some definitions for the currently known dichotomies for specific $\#\eo$ problems. $\mathscr{M}$ is the set of all ternary signatures whose supports only contain strings of Hamming weight exactly 1. In other words, a quaternary \eo\ signature $f\in \mathscr{M} \otimes \Delta_1 $, if and only if $\su(f) \subseteq \{1\} \times \{100,010,001\}=\{1100,1010,1001\}$. $\widetilde{\mathscr{M}}$ is defined similarly, except that  the Hamming weights of the supports are exactly 2. We also define

$$\mathscr{M}_{\mathscr{A}}=\{f\in\mathscr{M}\mid\text{the quotient of any two non-zero values of } f \text{ belongs to } \{\pm1,\pm\mathfrak{i}\}\}.$$ $\mathscr{M}_{\mathscr{A}} \otimes \Delta_1$ can be seen as $\mathscr{M} \otimes \Delta_1$ with the $\mathscr{A}$ type requirement. 
Similarly, we define

$$\widetilde{\mathscr{M}_{\mathscr{A}}}=\{f\in\widetilde{\mathscr{M}}\mid\text{the quotient of } \text{any two } \text{non-zero values } \text{of f } \text{belongs to } \{\pm1,\pm\mathfrak{i}\}\}.$$ 

With these definitions, we present the dichotomy for the six-vertex model \cite{cai2018complexity}, along with its generalized form that is $\#\eo$ problems defined by a set of binary or quaternary signatures \cite{meng2024p}.

\begin{theorem}[\cite{cai2018complexity}] \label{thm:sixvertexdichotomyMform}
    Let f be a quaternary $\eo$ signature, then $\hol(\neq_2\mid f)$ is \#P-hard except for the following cases:

a. $f\in\mathscr{P}$ (equivalently, $f \in \eom$ and $\tau(f) \subseteq \mathscr{P}$);

b. $f\in\mathscr{A}$ (equivalently, $f \in \eom$ and $\tau(f) \subseteq \mathscr{A}$);

c. $f\in \mathscr{M}\otimes  \Delta_1 $;

d. $f\in  \widetilde{ \mathscr{M}} \otimes  \Delta_0$;

in which cases $\hol(\neq_2\mid f)$ is computable in polynomial time.
\end{theorem}

\begin{theorem}[\cite{meng2024p}]\label{thm:arity4setdichotomy}
Suppose $\mathcal{F}$ is a set of Eulerian signatures of arity less than or equal to 4, mapping Boolean inputs to complex numbers. Then $\hol(\neq_2\mid\mathcal{F})$ is \#P-hard unless one of the following conditions holds:

a. $\mathcal{F}\subseteq \mathscr{P}\cup\mathscr{M}\otimes\Delta_1$;

b. $\mathcal{F}\subseteq \mathscr{A}\cup\mathscr{M}_{A}\otimes\Delta_1$;

c. $\mathcal{F}\subseteq \mathscr{P}\cup\widetilde{\mathscr{M}}\otimes\Delta_0$;

d. $\mathcal{F}\subseteq \mathscr{A}\cup\widetilde{\mathscr{M}_{\mathscr{A}}}\otimes\Delta_0$,

in which cases the problem can be computed in polynomial time.
\end{theorem}

In addition, some progress has been made towards $\eo$ signatures with larger arity. A special property is called ARS. A signature $f$ satisfying ARS \cite{cai2020beyond} if and only if $f(\overline\alpha)=\overline{f(\alpha)}$, where $\overline{f(\alpha)}$ denotes the complex conjugation of $f(\alpha)$. The dichotomy for signatures with ARS is presented as follows.

\begin{theorem}[{\cite[Theorem 3.1]{cai2020beyond}} ]\label{thm: dichotomy for ARS}
Let $\mathcal{F}$ be a set of $\eo$ signatures satisfying ARS. Then $\#\eo(\mathcal{F})$ is \#P-hard unless $\mathcal{F}\subseteq\mathscr{A}$ or $\mathcal{F}\subseteq\mathscr{P}$, in which cases it can be computed in polynomial time.
\end{theorem}

Another typical kind of $\eo$ signatures is introduced in \cite{meng2024p}, called the pure signatures.

\begin{definition}\label{def:pure}
    An $\eo$ signature $f$ is pure-up, if $\text{Span}(f)\subseteq\eog$. A signature set $\mathcal{F}$ is a pure-up \eo\ signature, if each signature in it is pure-up. Similarly, An $\eo$ signature $f$ is pure-down, if $\text{Span}(f)\subseteq\eol$. A signature set $\mathcal{F}$ is pure-down, if each signature in it is a pure-down \eo\ signature.

    Both pure-up and pure-down signatures are collectively referred to as pure signatures.
\end{definition}

We denote the property of pure-up by the symbol $\allup$, which means for each odd integer $k$, the bitwise addition of any arbitrary $k$ strings from $\su(f)$ is in $\eog$. Similarly, the property pure-down is denoted by $\alldown$.

We now introduce a significant definition, which looses the restriction on supports but require the affine part to be more specific.

\begin{definition}
  Suppose $f$ is an arity $2d$ $\eo$ signature and $S\subseteq \eoe$. $f|_{S}$ is the restriction of $f$ to $S$, which means when $\alpha\in S$, $f|_{S}(\alpha)=f(\alpha)$, otherwise $f|_{S}(\alpha)=0$.
  
  If for any perfect pairing $P$ of Var$(f)$, $f|_{\eom[P]}\in\mathscr{A}$, then we say that $f$ is $\eom[\mathscr{A}]$.
  
  Similarly, if for any perfect pairing $P$ of Var$(f)$, $f|_{\eom[P]} \in \mathscr{P}$, then $f$ is $\eom[\mathscr{P}]$.
  \label{def:eoaeop}
\end{definition}

The following lemma can be easily verified by Definition \ref{def:eoaeop}, \ref{defa} and \ref{defp}.
\begin{lemma}
    Suppose $f$ is an \eo\ signature. If $f\in\mathscr{A}$, then $f$ is $\eom[\mathscr{A}]$. If $f\in\mathscr{P}$, then $f$ is $\eom[\mathscr{P}]$.
    \label{lem:ap to eoaeop}
\end{lemma}
Given these definitions, we present the following dichotomy for pure signatures.

\begin{theorem}[\cite{meng2024p}]
    Suppose $\mathcal{F}$ is a set of pure-up (or pure-down) $\eo$ signatures. Then $\#\eo(\mathcal{F})$ is \#P-hard unless all signatures in $\mathcal{F}$ are $\eom[\mathscr{A}]$ or all of them are $\eom[\mathscr{P}]$, in which cases it can be computed in polynomial time.
    \label{thm:puredichotomy}
\end{theorem}

Another concept named rebalancing is also defined in \cite{meng2024p}. In addition, a polynomial-time algorithm is also presented with some conditions. 

\begin{definition}
An $\eo$ signature $f$ of arity $2d$, is called \ba[0](\ba[1] respectively), when the following  recursive conditions are met.
\begin{itemize}
    \item $d=0$: No restriction.
    \item $d\ge 1$: For any variable $x$ in $X=\text{Var}(f)$, there exists a variable $y=\psi(x)$ different from $x$, such that for any $\alpha\in\{0,1\}^X$, if $\alpha_x=\alpha_y=0$($\alpha_x=\alpha_y=1$ respectively) then $f(\alpha)=0$. Besides, the arity $2d-2$ signature $f^{x=0,y=1}$ is 0-rebalancing($f^{x=1,y=0}$ is \ba[1] respectively).
\end{itemize}
For completeness we view all nontrivial signatures of arity 0, which is a non-zero constant, as \ba[0](\ba[1]) signatures. Moreover, an \eo\ signature set $\mathcal{F}$ is said to be \ba[0](\ba[1] respectively) if each signature in $\mathcal{F}$ is \ba[0](\ba[1] respectively). 
\label{def:reba}
\end{definition}

\begin{theorem}[\cite{meng2024p}]
    The $\#\eo$ problem, defined over a finite set $\mathcal{F}$ of Boolean signatures, is solvable in polynomial time if each $f \in \mathcal{F}$ is 0-rebalancing (or 1-rebalancing), and all signatures in $\mathcal{F}$ belong to $\eom[\mathscr{A}]$ or all belong to $\eom[\mathscr{P}]$.
    \label{thm:rebaeasy}
\end{theorem} 

In this paper, we aim to make a complete complexity classification for complex weighted $\#\eo$ problems, based on the hypothesis that $\text{FP}^{\text{NP}}\neq \#$P. 
We remark that the algorithm for rebalancing signatures are covered by this dichotomy,
however the algorithm to be present in Section \ref{section: dichotomy for eoc} needs a specific NP oracle.

\section{Detailed version of Theorem \ref{thm:dicoceo}}\label{secdetail}
We first introduce some definitions of great significance, which is also needed in the detailed version of our main result.
\begin{definition}
    For an $\eo$ signature $f$, if there exists $\alpha,\beta,\gamma\in \su(f)$ and $\delta=\alpha\oplus\beta\oplus\gamma$, such that $\delta\in\eoe$ and $\delta\notin \su(f)$, then we say $f$ is a $\exists3\nrightarrow$ signature. If $\delta\in \eosg$ (or $\delta\in \eosl$) instead, we say $f$ is a $\exists 3\uparrow$ signature (or a $\exists 3\downarrow$ signature). If $f$ is neither a $\exists3\nrightarrow$ signature nor a $\exists 3\uparrow$ signature (or a $\exists 3\downarrow$ signature), we say it is a $\forall 3 \thdown$ signature (or a $\forall 3 \thup$ signature).
    \label{def:thupthdown}
\end{definition}

We remark that these definitions also motivate us to denote the pure-up signature as a $\allup$ signature. Besides, a signature is both $\forall 3 \thdown$ and $\forall 3 \thup$ if and only if it is an $\eom$ signature by definition. Our main result is presented as follows.
 \begin{theorem}\label{thmdetail}
      Let $\mathcal{F}$ be a set of \eo\ signatures. Then $\ceo(\mathcal{F})$ is \#P-hard, unless all signatures in $\mathcal{F}$ are $\mitsuup$ signatures or all signatures in $\mathcal{F}$ are $\mitsudown$ signatures, and $\mathcal{F}\subseteq \eom[\mathscr{A}]$ or $\mathcal{F}\subseteq \eom[\mathscr{P}]$, in which cases it is in $\text{FP}^{\text{NP}}$.
      
 In particular, if the aforementioned condition holds, then $\ceo(\mathcal{F})$ is polynomial-time computable, unless $\mathcal{F}$ is neither \ba[0] nor \ba[1].
 \end{theorem}

We remark that Theorem \ref{thmdetail} also presents the largest class of signature sets that induce a polynomial-time algorithm in \ceo\ and have been identified to date, which is originally established in \cite{meng2024p}. Using Theorem \ref{thmdetail}, two more dichotomies for specific signature sets can be obtained as corollaries, stated as Corollary \ref{coro:upside} and \ref{coro:sw}. The proof of \ref{thmdetail} involves two theorems proved in Section \ref{sec:ceo} and \ref{section: dichotomy for eoc} respectively, stated as follows.
We use $\ceoc(\mathcal{F})$ to denote $\ceo(\mathcal{F}\cup\{\Delta\})$.

\begin{theorem}
    Suppose $\mathcal{F}$ is a set of \eo\ signatures. Then one of the following holds:
    \begin{enumerate}
        \item All signatures in $\mathcal{F}$ are  \eo\ signatures satisfying ARS, ignoring a constant;
        \item $\ceoc(\mathcal{F})\equiv_T \ceo(\mathcal{F})$.
    \end{enumerate}
    \label{thm:eoc to eo}
\end{theorem}

\begin{theorem}[Dichotomy for $\#\eo^c$]
      Let $\mathcal{F}$ be a set of \eo\ signatures. Then $\ceoc(\mathcal{F})$ is \#P-hard, unless all signatures in $\mathcal{F}$ are $\mitsuup$ signatures or all signatures in $\mathcal{F}$ are $\mitsudown$ signatures, and $\mathcal{F}\subseteq \eom[\mathscr{A}]$ or $\mathcal{F}\subseteq \eom[\mathscr{P}]$, in which cases it is in $\text{FP}^{\text{NP}}$.
  \label{thm:dicoceoc}
 \end{theorem}
\begin{proof}[Proof of Theorem \ref{thmdetail}]
     The polynomial-time computable result is obtained from Theorem \ref{thm:rebaeasy}. Now we focus on the $\pnp$ vs. \#P dichotomy. The $\pnp$ algorithm is given by Lemma \ref{lem:NPalg}. By Theorem \ref{thm:eoc to eo},  one of the following cases holds.
    
   \textbf{ Case 1: }All signatures in $\mathcal{F}$ are  \eo\ signatures satisfying ARS, ignoring a constant. If further $\mathcal{F}\subseteq\mathscr{A}$ or $\mathcal{F}\subseteq\mathscr{P}$, then all signatures in $\mathcal{F}$ have affine supports. By Lemma \ref{lem: eo+affine= eom}, all signatures in $\mathcal{F}$ are $\eom$, and consequently are both $\forall 3 \thdown$ and $\forall 3 \thup$. Furthermore,  all signatures in $\mathcal{F}$ are $\eom[\mathscr{A}]$ or $\eom[\mathscr{P}]$ by Lemma \ref{lem:ap to eoaeop}. Consequently, if $\mathcal{F}$ does not satisfy the criteria for the $\pnp$ algorithm, then  $\mathcal{F}\nsubseteq\mathscr{A}$ and $\mathcal{F}\nsubseteq\mathscr{P}$ by the inverse negative proposition. In this case we are done by Theorem \ref{thm: dichotomy for ARS}.    
   
     \textbf{ Case 2: }$\ceoc(\mathcal{F})\equiv_T \ceo(\mathcal{F})$. Then we are done by Theorem \ref{thm:dicoceoc}.
\end{proof}

\section{Reduce $\ceoc$ to $\#\eo$}\label{sec:ceo}
In this section, we prove Theorem \ref{thm:eoc to eo}. 
It should be noted that in the subsequent analysis, the proof sometimes demonstrates that $\ceo(\mathcal{F})$ is \#P-hard. In this case, $\ceoc(\mathcal{F})\equiv_T \ceo(\mathcal{F})$ holds by the trivial reduction.

Firstly, we define a specific generating process to construct binary signatures with a single signature $f$ in Section \ref{sec:generateprocess}. These generated binary signatures can be regarded as materials to realize $\Delta$. Based on what kind of binary signatures can be generated, our proof is then separated into three parts.

\subsection{A specific generating process}\label{sec:generateprocess}


We begin with defining the normalization for a non-zero binary signature. We can suppose without loss of generality that a nonzero binary signature is of the form $\bi{a}{b}$, where $a\neq0$ and $|a|>|b|$. Except for changing the partition function by a constant, it works equivalently as $\bi{1}{r}$, where $r=\frac{b}{a}$ and $|r|\leq1$. As a result, $\bi{1}{r}$ is regarded as the \textit{normalized form} of $\bi{a}{b}$. In this subsection, a root of unity is briefly called a root. Notably, the whole setting is under \ceo problems, which implies that all connecting operations are using $\neq_2$. We use $n$ to denote the size of the instance $I$.

\begin{definition}[Generating process] \label{def: generating process}
Given a signature $f$ of arity $2d$, we recursively define a sequence of sets:

    I. $B_0(f)=\{\neq_2\}$.
    
    II. Given $B_{i-1}(f)$, $A_i(f)$ is composed of all binary signatures, realized by adding $d-1$ self-loops on arbitrary $2d-2$ variables of $f$ by signatures arbitrarily chosen from $B_{i-1}(f)$.
   
    III. Given $A_i(f)$, $B_i(f)$ is defined as $\{ \bi{1}{\prod_{1\leq j\leq k}r_j}  \mid \bi{1}{r_j}\in A_i(f)\textsf{ and }  k \in \mathbb{Z}_+ \}$. Signatures in  $B_i(f)$ can be realized by connecting $k$ signatures from  $A_i(f)$ as a path.

Let $B(f)=\mathop{\bigcup}\limits_{i=1}^{\infty}B_i$. The whole process is called the generating process of $f$.
\end{definition}

Based on the generating process, the proof is structured as follows. 
\begin{enumerate}
    \item
    If a binary signature $\bi{1}{r}$ can be generated by some $f\in\mathcal{F}$, where $r$ is not a root of unity, then we realize $\Delta$ by the standard polynomial interpolation reduction in Section \ref{subsec:EOC part: nonroot-interpolation}.
    
    \item
    If we can generate infinite binary signatures with some $f\in\mathcal{F}$, and each of them can be normalized to $\bi{1}{r}$ with $r$ as a root of unity, we will analyze the efficiency of the generating process and evaluating the size of gadget constructions, then operate the polynomial-time interpolation reduction in Section \ref{subsection: infinitely many roots}. 

     \item
     If we can only generate a finite number of different binary signatures for each $f\in \mathcal{F}$, and each generated signature can be normalized to $\bi{1}{r}$ with $r$ as a root of unity, then we make classifications and deal with each case in Section \ref{subsec: finitely many roots}.

 \end{enumerate}
We remark that for the second case, the generating process may not stop, but we will show that we only need polynomial time to generate all the signatures we need for the interpolation for $I$.
Furthermore, if $B(f)$ is finite, then there exists a constant $i\ge 0$ such that $B_i(f)=B_{i+1}(f)$, and we stop the generating process at the $(i+1)$th step in $O(1)$ time. In this case, we then choose the next signature from $\mathcal{F}$ and perform the generating process. Consequently, given a signature set $\mathcal{F}$, we can always stop the generating process in polynomial time, and decide which statement of Theorem \ref{thm:eoc to eo} we apply.

\subsection{Not a root of unity}\label{subsec:EOC part: nonroot-interpolation}


Suppose $x$ is not a root and $\bi{1}{x}$ is generated in the $i$th step by some $f\in\mathcal{F}$. Therefore, given $f$, we can realize $\bi{1}{x}$ of constant size, hence $\#\eo(\mathcal{F}\cup\{\bi{1}{x}\})\le_T\#\eo(\mathcal{F})$. In the following, we prove that $\ceoc(\mathcal{F})\le_T\#\eo(\mathcal{F}\cup\{\bi{1}{x}\})$. If $x=0$, the reduction trivially holds. Now we assume $x\neq0$.

Suppose $I$ is an instance of $\ceoc(\mathcal{F})$ with $n$ vertices, and $\pin$ appears $m$ times in $I$, where $m\leq n$. When we replace each $\pin$ in $I$ with $\bi{1}{a}$, where $a$ is a complex variable, Z$(I,a)$ can be viewed as an $m$-arity polynomial of $a$. In other words, Z$(I,a)=\sum_{i=0}^{m}c_ia^i$, in which $c_i$ is the sum of the evaluations in which there are $i$ pairs of variables assigned the value to make $\bi{1}{a}$ contribute an $a$ to Z$(I)$. In particular, Z$(I)=\text{Z}(I,0)$.

By connecting $i$ copies of $\bi{1}{x}$ we can get $\bi{1}{x^i}$, $i\in\{1,2,...,m+1\}$. After replacing each $\pin$ with $\bi{1}{x^i}$, the number of vertices in the modified instance are $O(n^2)$, which is a polynomial of $n$. Then we use these $m+1$ binary signatures to replace $\pin$ in $I$ respectively, which actually forms a linear system
$$\left( \begin{array}{cccc}
x^0 & x^1 & ... & x^m\\
x^{1\times0} & x^{1\times1} & ... & x^{1\times m}\\
\vdots & \vdots & \ddots & \vdots\\
x^{m\times0} & x^{m\times1} & ... & x^{m\times m}
\end{array} 
\right)\left( \begin{array}{c}
c_0\\
c_1\\
\vdots\\
c_m
\end{array} 
\right)=\left( \begin{array}{c}
\text{Z}(I,x)\\
\text{Z}(I,x^2)\\
\vdots\\
\text{Z}(I,x^m)
\end{array} 
\right).$$ 
We denote this linear system as $\textbf{AX}=\textbf{Z}$ for short. Using the oracle of ${\#\eo}(f\cup\{\bi{1}{x}\})$, we know the value of $\textbf{Z}$. Furthermore, the coefficient matrix $\textbf{A}$ is Vandermonde since $x$ is not a root and all $x^i$'s are distinct. As $\textbf{A}$ is invertible, we can solve $\textbf{X}$ in polynomial time. In particular, Z$(I)=c_0$. As a result,
$\textsf{\ceo}^c(\mathcal{F})\leq_T\ceo(f\cup\{\bi{1}{x}\})\le_T\ceo(\mathcal{F}).$ It's obvious that $\ceo(\mathcal{F})\leq_T\ceoc(\mathcal{F})$, hence we have $\ceoc(\mathcal{F})\equiv_T\ceo(\mathcal{F}).$

\subsection{Infinitely many roots of unity}\label{subsection: infinitely many roots}

The content of this subsection happens to be highly related to the computation models. We study complex weighted version of this problem, to capture nontrivial, general forms of certain tractable classes. However, we encounter with a gap between the standard discrete Turing Machine model and the representation and computation of complex numbers. Rather than addressing this gap in detail, we outline two approaches to overcome this difficulty.

One approach is to consider an finite extension field of $\mathbb{Q}$, which includes all weights appearing in the problem, following the discussion on models of computation in \cite[Section 1.4]{cai2017complexity}. In the commonly used computational model, we can restrict all possible weights to algebraic numbers instead of the entire field of complex numbers. In this case, as the function $f$ involves only finitely many algebraic numbers, the relevant extension field of the problem $\#\eo(f)$ would be a finite algebraic extension containing all possible weights of all gadgets. This restriction prevents the occurrence of infinitely many roots. Consequently, when focusing on algebraic numbers, the problem naturally fits in the standard discrete Turing Machine model, and further discussion within this subsection becomes unnecessary.

Alternatively, we may consider computational models capable of handling real numbers directly, representing complex numbers as two-dimensional vectors. For instance, the Real RAM model as discussed in \cite{preparata2012computational} could be used for computations involving real numbers. For our purposes in this subsection, a “Complex RAM” model would be required, where basic operations on complex numbers, such as addition, multiplication, and the logarithm function, are treated as single computational steps. Assume such model is available, in the following we will finish the reduction from the perspective of mathematics, avoiding any further discussions on the computational model.

The general reduction framework also relies on polynomial interpolation, similar to the process described in Section \ref{subsec:EOC part: nonroot-interpolation}.  We also need $N=n+1$ different binary signatures, where $n$ is the size of instance given. We only have to show how to get them efficiently in the current situation.

Now there are infinitely many binary signatures that can be generated. All of them have parameters as roots of unity after normalization, that is, they are all of the form $\bi{1}{x}$, where $x$ is a root of unity.

We first characterize the structure of all sets generated. By Definition \ref{def: generating process} we have that when there is an $i$ satisfying $B_i=B_{i+1}$, then all $B_j$ with $j>i$ are the same. At this situation, the number of signatures generated is finite, which can be covered by Section \ref{subsec: finitely many roots}. Therefore, for all $i\in\mathbb{N}_+$, $B_i\subsetneqq B_{i+1}$. In addition, since there are finitely many ways to use $B_i$ to construct $A_{i+1}$ and all members of $A_{i+1}$ are parameterized by roots of unity, by spiral induction we have all $A_i$ and $B_i$ are finite sets.

For preparation, we introduce a lemma that characterize how roots of unity interact.

\begin{lemma}\label{lem: roots combination}
Assume $a$ and $b$ are relatively prime positive integers, for any given $t$, there exist $0\leq r< a$ and $0\leq s< b$, such that $\{\frac{r}{a}+\frac{s}{b}\}=\{\frac{t}{ab}\}$. In the statement $\{\cdot\}$ denotes the decimal part function. Additionally, $r$ and $s$ can be found in $O(ab)$ time.
\end{lemma}
\begin{proof}
Without loss of generality, we can assume that $0\leq t<ab$, then the equation becomes $$\{\frac{r}{a}+\frac{s}{b}\}=\frac{t}{ab}$$ After reduction of fractions and with the basic arithmetic property of $\{\cdot\}$ we have 
\begin{equation}\label{equation: tongyushi}
    as+br \equiv t \text{ (mod }ab)
\end{equation}
With Bézout's identity and the method of successive division, we can get the answer $r'$ and $s'$ for the equation $as+br=t$, in which $r'$, $s'$ are integers, in $O(ab)$ time. Here the time complexity can be further optimized, but $O(ab)$ is sufficient for our analysis. Then we adjust these two answers to what we want by letting $r''=r'+k_1a$ and $s''=s'+k_2b$, in which $k_1$, $k_2$ are integers. Noticing that for any $k_1$ and $k_2$, $r''$ and $s''$ are answers to \ref{equation: tongyushi}, therefore, we can assign proper values to them to make $0\leq r< a$ and $0\leq s< b$. The adjusting also can be done in $O(ab)$ time. 
\end{proof}

\begin{remark}\label{remark: rc/a+sd/b=t/ab}
We use gcd to denote the greatest common divisor and lcm to denote the least common multiple. Noticing that when $\text{gcd}(p,q)=1$, $\{\frac{0}{p},\frac{1}{p},\ldots,\frac{p-1}{p}\}=\{\frac{0\cdot q}{p},\frac{q}{p},\ldots,\frac{(p-1)q}{p}\}$. Therefore, as a corollary of lemma \ref{lem: roots combination}, for any given $t$, there exist $0\leq r< a$ and $0\leq s< b$, such that $\{\frac{rc}{a}+\frac{sd}{b}\}=\{\frac{t}{ab}\}$, where $\text{gcd}(a,c)=\text{gcd}(b,d)=1$. $r$ and $s$ can also be found in $O(ab)$ time. 
\end{remark}

\begin{remark}\label{remark: one to one respond between s/b+r/a and t/lcmab}
Another thing we have to notice is that there are exactly $ab$ choices of $(r,s)$, which is the same number as the number of integers in the range of $[0,ab)$, so they are one-to-one correspondence. 

When $a$ and $b$ are not relatively prime, assume $b=b'd$, where $d=\text{gcd}(a,b)$. We can see that $\{\{\frac{s}{b}+\frac{r}{a}\}|0\leq r< a\}=\{\{\frac{s+b'}{b}+\frac{r}{a}\}|0\leq r< a\}$ is correct for all $0\leq s\leq b'-1$. And $\{\{\frac{s}{b}+\frac{r}{a}\}|0\leq r< a, 0\leq s\leq b'-1\}$ has exactly $ab'=\text{lcm}(a,b)$ members, therefore they are also one-to-one corresponding to $\{\frac{t}{\text{lcm}(a,b)}|0\leq t< \text{lcm}(a,b)\}$, which is important when we need to list all members in the latter set. 

Furthermore, $\{\{\frac{sd}{b}+\frac{r}{a}\}|0\leq r< a, 0\leq s\leq b'-1\}$ also has exactly $ab'=\text{lcm}(a,b)$ different terms, where $\text{gcd}(b,d)=1$. Suppose otherwise, there exist $\{\frac{sd}{b}+\frac{r}{a}\}=\{\frac{s'd}{b}+\frac{r'}{a}\}$, where $0\leq r,r'< a, 0\leq s,s'\leq b'-1$ and $(r,s)\neq (r',s')$. Then $(s-s')da'=(r'-r)b'+k\text{lcm}(a,b)$ is correct for some integer $k$. We have $b'\mid (s-s')$ so $s-s'$ must be 0. Then $a\mid (r'-r)$ so $r'-r=0$, which causes a contradiction.
\end{remark}

 We now make a basic assumption that there is a $f\in\mathcal{F}$ with the property that by the generating process it can generate infinitely many binary signatures, parameterized by roots of unity. The reason why we can make this assumption without loss of generality is our main concern is to restrict the increasing size of gadgets, so we won't worry about using signatures other than $f$ to generate binary signatures. Because every time we use a new signature in $\mathcal{F}$ to generate signatures, the size of gadgets begin to increase from 1 all over again, as is implied by the proof in the rest of this subsection.

Assume the arity of $f$ is $2d$. The generating process gives a series of finite sets $B_0$, $A_1$, $B_1$, $A_2$, $B_2$, $\dots$. Parameters of all signatures in those sets are roots of unity. For convenience, in the rest of this subsection, we briefly denote $\bi{1}{e^{\frac{2\pi q\mathfrak{i}}{p}}}$ as $\frac{q}{p}$, without causing ambiguity. 

By Definition \ref{def: generating process} and Lemma \ref{lem: roots combination} we have that for all $i\in\mathbb{N}$, $B_i$ contains all binary signatures parameterized by roots of unity with the same order. By the generator of one set we mean signatures that can be combined to generate the whole set. We use the generator of $B_i$ to denote it, such as $B_i=(1,\frac{q_1}{p_1},\ldots,\frac{q_{b_i-1}}{p_{b_i-1}})$, where $b_i$ is the number of generators. Naturally, $B_i$ contains all signatures parameterized by roots of unity of the order $\text{lcm}(p_1,\ldots,p_{b_i-1})$.

We modify the generating process for clearer analysis. We change the II step into:
Given $B_{i-1}=(1,\frac{q_1}{p_1},\ldots,\frac{q_{b_{i-1}-1}}{p_{b_{i-1}-1}})$, where $i\in\mathbb{N}$ and $i>0$, we enumerate all possible ways to add $(d-1)$ self-loops on $f$, until a binary signatures $\bi{1}{\frac{q}{p}}$ satisfying $p\nmid \text{lcm}(p_1,\ldots,p_{b_{i-1}-1})$ is found. We add $\bi{1}{\frac{q}{p}}$ to $B_{i-1}$, and call it $A_i$. The other part of the process remains unchanged. Intuitively speaking, instead of finding all new binary signatures using $B_{i-1}$, we find one and stop. It can be verified that following the modified generating process, the number of generators increases exactly by 1 from $B_{i-1}$ to $B_i$, that is $b_i=i+1$.

Notably, in the new generating process we delay the time of the appearance of some binary signatures, which will only cause the size of gadgets to be larger, therefore making this modification is reasonable. This statement will also be illustrated by the following computation of the size of gadgets.

Now we begin to analyze the size of gadgets. Through the generating process, Suppose $B_i=(1,\frac{q_1}{p_1},\ldots,\frac{q_i}{p_i})$, which implies that $B_{i}$ contains all roots of unity with the order $\text{lcm}(p_1,\ldots,p_i)$ and the new signature added to $A_i$ is $\frac{q_i}{p_i}$. Suppose any signatures in $B_i$ can be represented by a gadget with at most $n_i$ vertices, and $A_i$ with at most $m_i$ vertices. There exists an $i_0$ such that $B_{i_0}$ contains more than $n$ roots of unity. Since $p_i\nmid \text{lcm}(p_1,\ldots,p_{i-1})$, we have $i_0\leq log_2^{n}+2$. 

For all $1\leq i\leq i_0$, we let $N_i=\text{lcm}(p_1,\ldots,p_i)$ that is also the size of $B_i$. Naturally we have $N_{i_0}\geq n$ and $N_i=N_{i-1}p_i/\text{gcd}(N_{i-1},p_i)$. We regard $p_i/\text{gcd}(N_{i-1},p_i)$ as a whole named $p_i'$ and particularly $p_1'=p_1$.  For any arbitrary signature $\frac{k}{N_i}$ in $B_i$, where $0\leq k<N_i$, by repeatedly using the third paragraph of \ref{remark: one to one respond between s/b+r/a and t/lcmab} from the backward, we know that for each signature we need at most $p_j'$ copies of the $p_j$ signature for all $1\leq j\leq i$ combined to gadget construct it.

By the analysis above, we have $m_i\leq (d-1)n_{i-1}+1$ and $n_i\leq \sum_{j=0}^{i}m_jp_{j}'$. Let $n_0=1$ and $p_0'=p_0=1$. We now use induction to prove that $n_i\leq \sum_{j=0}^{i} d^j\prod_{k=0}^{j}p_k'$ and $m_i\leq d^{i}\prod_{k=1}^{i-1}p_k'$. It can be verified that when $i=1$ the statement is correct. Assume $n_{i-1}$ satisfies the disequality, then

$$m_{i}\leq (d-1)n_{i-1}+1=(d-1)\sum_{j=0}^{i-1} d^{j}\prod_{k=0}^{j}p_k'+1\leq d^i\prod_{k=0}^{i-1}p_k',$$ and 

$$n_i\leq \sum_{j=0}^{i}m_jp_{j}'\leq \sum_{j=0}^{i} d^j\prod_{k=0}^{j}p_k'.$$ Then the spiral induction is done. Furthermore, if $i<i_0$, we have $m_i\leq nd^{i}.$

Since $i_0\leq log_2^{n}+2$ and $i_0$ is the smallest integer such that $|B_{i_0}|>n$, all members in $A_{i_0}$ can be constructed by less than $d^{i_0}\prod_{k=0}^{i_0-1}p_k'\leq nd^{i_0}$ vertices, which are all poly($n$)-size. For the newly constructed signature $p_{i_0}\in A_{i_0}$, we first let there are $0$ copies of $p_{i_0}$, then we can get $N_{i_0-1}$ independent signatures using poly(n)-size gadget constructions. Then we let there are $1$ copy, and so on, until we get $n+1$ signatures. According to Lemma \ref{lem: roots combination} and Remark \ref{remark: one to one respond between s/b+r/a and t/lcmab}, we get independent $n+1$ signature with this method. The number of vertices in their gadget constructions are smaller than $nm_{i_0}+n_{i_0-1}\leq n^2d^{i_0}+\sum_{j=0}^{i_0-1} d^j\prod_{k=0}^{j}p_k'< n^2d^{i_0}+d^{2i_0}$. It's still poly(n)-size.

Noticing that the enumeration of $A_i$ by $B_{i-1}$ are essentially listing poly($n$) members, connecting them and comparing them, under the computation model with the logarithm function treated as single computation steps. Combining all analysis above, the whole process can be done in poly($n$) time. Therefore, combined with polynomial interpolation method, at this situation, we get 
$\ceoc(\mathcal{F})\equiv_T\ceo(\mathcal{F})$.

\subsection{Finitely many roots of unity} \label{subsec: finitely many roots}


Before the main part of this subsection, we give the definition of a specific $\#\eo$ problems, where some of its signatures can only appear a limited number of times. We then prove some useful reduction results based on it. It's noticeable that this definition is an extension of \cite[Definition 8.8]{shao2020realholant}.

\begin{definition}
    Let $\mathcal{F}$ be a set of $\eo$ signatures and $f$ be a $\eo$ signature. The problem $\#\eo(\mathcal{F},f^{\leq k})$ contains all instances of $\#\eo(\mathcal{F}\cup\{f\})$ where the signature $f$ appears at most $k$ times. 
\end{definition}


\begin{lemma}
    Let $\mathcal{F}$ be a finite set of $\eo$ signatures. We have

$$\#\eo(\mathcal{F},\pin^{\leq1})\leq_T\#\eo(\mathcal{F}).$$
\label{lem:single pin}
\end{lemma}
\begin{proof}

   Suppose $I_1$ is an instance of $\#\eo(\mathcal{F},\pin^{\leq1})$. If $\pin$ does not appear in $I_1$, the reduction trivially holds. Now suppose $\pin$ appears exactly once, and it is incident to $x_1$ and $x_2$ via $\neq_2$. That is to say the input of $\pin$ is $\overline{x_1},\overline{x_2}$. Also let $\pin(\overline{x_1},\overline{x_2})=1$ when $x_1=0,x_2=1$. In other words, $(x_1,x_2)$ is pinned to $(0,1)$ in $I_1$. We denote the partition function $Z(I_1)$ by $Z_1$.
   
   Now we construct another instance $I_2$ by reversing the direction of $\pin$ in $I_1$. In other words, $I_2$ is same as $I_1$ except that $\pin(\overline{x_1},\overline{x_2})=1$ when $x_1=1,x_2=0$.  We denote the partition function $Z(I_2)$ by $Z_2$. We then construct an instance $I_3$ by removing $\pin$ from $I_1$ and connect $x_1,x_2$ via $\neq_2$.  We denote the partition function $Z(I_3)$ by $Z_3$. By Definition \ref{def:numbereo} we have $Z_3=Z_1+Z_2$. In addition, by removing $\pin$ from $I_1$, we construct a binary $\mathcal{F}$-gate whose signature is $\neq_2^{Z_1,Z_2}$.
   
   If the signatures of all binary $\mathcal{F}$-gates are of the form $\lambda\neq_2,\lambda\in\mathbb{C}$, then we have $Z_1=Z_2$ and consequently $Z_1=Z_3/2$. As $I_3$ is an instance of $\ceo(\mathcal{F})$, the reduction holds. Consequently we may assume there exists a binary $\mathcal{F}$-gate whose signature is $\neq_2^{a,b},a\neq b$ .  We also remark that the size of this $\mathcal{F}$-gate is bounded by a constant due to the existence of the corresponding gadget construction. Now we construct another instance $I_4$ by replacing $\pin$ in $I_1$ with $\neq_2^{a,b}$. We denote the partition function $Z(I_4)$ by $Z_4$ and we have $Z_4=aZ_1+bZ_2$. Since $a\neq b$, together with $Z_3=Z_1+Z_2$ we can solve the linear equations to compute $Z_1$.  As $I_4$ is also an instance of $\ceo(\mathcal{F})$, the reduction holds. 
\end{proof}

\begin{remark}\label{remark:decidableexistence}
     In fact, deciding whether all binary $\mathcal{F}$-gates are of the form $\lambda\neq_2,\lambda\in\mathbb{C}$ can be done in constant time. We will explain this in Section \ref{subsection:generating 2 signature}.
\end{remark}

\begin{lemma}[{\cite[Lemma 5.1]{cai2020beyond}}]
    $\#\eo(\mathcal{F}\cup\mathcal{DEQ})\leq_T\#\eo(\mathcal{F}\cup\{\neq_4\})$.
    \label{lem:neq4 to deq}
\end{lemma}
\begin{lemma}\label{ouhetoeoc}
    Let $\mathcal{F}$ is a finite set of $\eo$ signatures and $\ouhe\in\mathcal{F}$. Then
    $\ceoc(\mathcal{F})\leq_T\#\eo(\mathcal{F}).$
\end{lemma}
\begin{proof}
    By Lemma \ref{lem:single pin} and \ref{lem:neq4 to deq} we have $\#\eo(\mathcal{F}\cup\mathcal{DEQ},\pin^{\leq1})\le_T\#\eo(\mathcal{F}\cup\mathcal{DEQ})\leq_T\#\eo(\mathcal{F}\cup\{\neq_4\})$. Now we realize unlimited $\pin$ by a single $\pin$ and a signature in $\mathcal{DEQ}$. Given an instance of $\ceoc(\mathcal{F})$, assume $\pin$ has appeared $k$ times. Then by properly connecting two edges of a $\pin$ to a $\neq_{2k+2}$, we can realize the signature $\pin^{\otimes k}$, which meets our requirements. Therefore we have $\ceoc(\mathcal{F})\leq_T\#\eo(\mathcal{F})$.
\end{proof}

We also need the following lemma in our proof, whose idea comes from \cite{lin2018complexity}.
\begin{lemma}
    Let $f$ be a signature and $\mathcal{F}$ be a set of signatures. Then for any $d\in \mathbb{N}^+$,

    $$\hol(\neq_2\mid\mathcal{F}\cup\{f\})\equiv_T \hol(\neq_2\mid\mathcal{F}\cup\{f^{\otimes d}\})$$
    \label{lem:linwang}
\end{lemma}
\begin{proof}
    We prove this by induction. The lemma is trivial when $d=1$. Assume the lemma holds for all $1\le d <k$. Now we  consider the case when $d=k$. In the following, we prove $\hol(\neq_2\mid\mathcal{F}\cup\{f\})\le_T \hol(\neq_2\mid\mathcal{F}\cup\{f^{\otimes k}\})$ as 
    $\hol(\neq_2\mid\mathcal{F}\cup\{f^{\otimes k}\})\le_T \hol(\neq_2\mid\mathcal{F}\cup\{f\})$ is trivial.

    Suppose there exists an instance $I$ of $\hol(\neq_2\mid\mathcal{F}\cup\{f\})$ such that Z$(I)\neq 0$ and $f$ appears $ak-b$ times in $I$, where $a\in \mathbb{N}^+$ and $0<b<k$ are two integers. Using $a$ many $f^{\otimes k}$ and signatures from $\mathcal{F}$, we may construct the instance $I$ with $b$ many $f$ unused. This actually forms a gadget on RHS, whose signature is $f^{\otimes b}$ multiplied by Z$(I)$. Then by induction we have $\hol(\neq_2\mid\mathcal{F}\cup\{f\})\le_T \hol(\neq_2\mid\mathcal{F}\cup\{f^{\otimes b}\})\le_T \hol(\neq_2\mid\mathcal{F}\cup\{f^{\otimes k}\})$ since $0<b<k$.

    Otherwise, for any instance $I$ of $\hol(\neq_2\mid\mathcal{F}\cup\{f\})$, if $f$ doesn't appear a multiple of $k$ times, Z$(I)=0$ and the reduction is trivial. And if $f$ appears a multiple of $k$ times, we simply replace each $k$ many $f$ in $I$ with a $f^{\otimes k}$ signature. This implies that $\hol(\neq_2\mid\mathcal{F}\cup\{f\})\le_T \hol(\neq_2\mid\mathcal{F}\cup\{f^{\otimes k}\})$.
\end{proof}
We also present a lemma to deal with a special case in our analysis.
\begin{lemma}\label{lem:self-loop to nontrivial eo}
Suppose $f$ is a non-trivial \eo\ signature of arity $2d> 2$ and $\alpha\in\su(f)$. Then there exists a signature $g$ obtained by adding a self-loop such that $g$ is non-trivial.
\end{lemma}
\begin{proof}
    Without loss of generality, we write $\alpha$ as $\alpha=100\delta$. Consider $\beta=010\delta$ and $\gamma=001\delta$. Let $f(\alpha)=x$, $f(\beta)=y$ and $f(\gamma)=z$. If $x+y\neq0$, we add a self-loop on $x_1$ and $x_2$ to realize a signature $g$ such that $g(0\delta)\neq0$. If $x+z\neq0$ (or $y+z\neq0$), we similarly add a self-loop on $x_1$ and $x_3$ (or on $x_2$ and $x_3$) to realize $g$. If $x+z=x+y=y+z=0$, then $x=0$, which contradicts $\alpha\in\su(f)$. 
\end{proof}

Now we are ready to analyze the following cases.
\subsubsection{Generating more than 2 signatures}

We consider the case that there exists an integer $k\geq3$ such that $B(f)=\{\bi{1}{\omega}\mid\omega^k=1\}$, where $f\in\mathcal{F}$. We commence with the property of such $f$.
\begin{lemma}\label{complementary equal k}
    Assume $f$ is a signature with even arity and $k\geq3$ is an integer. If $B(f)=\{\bi{1}{\omega}\mid\omega^k=1\}$, then for arbitrary $\alpha,\beta\in\su(f)$, $f(\alpha)^k=f(\overline{\alpha})^k$,   $f(\alpha)f(\overline\alpha)=f(\beta)f(\overline\beta)$. 
\end{lemma}
\begin{proof}
    We actually prove a stronger statement: Assume $f$ is a signature with even arity and $k\geq3$ is an integer. If $B(f)\subseteq\{\bi{1}{\omega}\mid\omega^k=1\}$, and all signatures in $\{\bi{1}{\omega}\mid\omega^k=1\}$ are available (meaning that these signatures can be used in the generating process), then for arbitrary $\alpha,\beta\in\su(f)$, $f(\alpha)^k=f(\overline{\alpha})^k$,   $f(\alpha)f(\overline\alpha)=f(\beta)f(\overline\beta)$. 
    
    We prove this statement by induction. First, we prove that  $f(\alpha)^k=f(\overline{\alpha})^k$. We can easily verify that for binary \eo\ signatures the statement trivially holds. Assume for all $f$ of arity less than $2d$ the statement holds, where $d\geq2$. Now suppose $\text{arity}(f)=2d$.

Without loss of generality, suppose $\alpha=100\delta,\beta=010\delta,\gamma=001\delta$ are three strings with Hamming weight $d$. Suppose $f(\alpha)=a$, $f(\overline{\alpha})=a'$, $f(\beta)=b$, $f(\overline{\beta})=b'$, $f(\gamma)=c$, $f(\overline{\gamma})=c'$,. We add a self-loop between  $x_1$ and $x_2$ via $\bi{1}{\omega},\omega^k=1$  to realize $f_{12}$, then by the induction hypothesis we have $(a+\omega b)^k=f^k_{12}(0\delta)=f^k_{12}(\overline{0\delta})=(\omega a'+b')^k$ since $B(f_{12})\subseteq\{\bi{1}{\omega}\mid\omega^k=1\}$. By the Binomial Formula Expansion we have:

    $$(a^k+b^k-a'^k-b'^k)+\omega\left( \begin{array}{c}
k \\
1 
\end{array} 
\right)(a^{k-1}b-a'b'^{k-1})+\ldots+\omega^{k-1}\left( \begin{array}{c}
k \\
k-1 
\end{array} 
\right)(ab^{k-1}-a'^{k-1}b')=0.$$ Suppose $\omega_0=e^{\frac{2\pi\mathfrak i}{k}}$. 
With all possible assignments of $\omega$ we have:

$$\left( \begin{array}{ccccc}
1 & 1 & 1 & \ldots & 1 \\
1 & \omega_0 & \omega_0^2 & \ldots &\omega_0^{k-1} \\
1 & \omega_0^2 & \omega_0^4 & \ldots &\omega_0^{2(k-1)} \\
\vdots & \vdots & \vdots & \ddots & \vdots \\
1 & \omega_0^{k-1} & \omega_0^{2(k-1)} & \ldots &\omega_0^{(k-1)^2}
\end{array} 
\right)\left( \begin{array}{c}
a^k+b^k-a'^k-b'^k \\
\left( \begin{array}{c}
k \\
1 
\end{array} 
\right)(a^{k-1}b-a'b'^{k-1}) \\
\left( \begin{array}{c}
k \\
2 
\end{array} 
\right)(a^{k-2}b^2-a'^2b'^{k-2}) \\
\vdots \\
\left( \begin{array}{c}
k \\
k-1 
\end{array} 
\right)(ab^{k-1}-a'^{k-1}b')
\end{array} 
\right)=\mathbf{0}.$$
The coefficient matrix is a Vandermonde matrix, thus the system of linear equations has only one solution $\mathbf{0}$. We have $a^k+b^k=a'^k+b'^k$ and $a^{k-i}b^i=a'^ib'^{k-i}$, where $1\leq i\leq k-1$, $i\in\mathbb{N}$.

  Similarly by adding a self-loop between $(x_1,x_3)$ and $(x_2,x_3)$ respectively, we have $a^k+c^k=a'^k+c'^k$ and $b^k+c^k=b'^k+c'^k$.  Solving these linear equations we have $a^k=a'^k$. The choice of $\alpha$ is arbitrary, so for all $\alpha\in\su(f)$, $f(\alpha)^k=f(\overline\alpha)^k$, the first property is proved.

Then, we prove that $f(\alpha)f(\overline\alpha)=f(\beta)f(\overline\beta)$. We prove this property by introducing another induction. If $\alpha,\beta\in\su(f)$, then $a,b\neq0$. Since $k\ge 3$, we have $a^{k-1}b=a'b'^{k-1}$ and $a^{k-2}b^2=a'^2b'^{k-2}$, and consequently $aa'=bb'$. Here we actually proved that for $\alpha,\beta\in\su(f)$ with $\dis{\alpha}{\beta}=2$, $f(\alpha)f(\overline\alpha)=f(\beta)f(\overline\beta)$, which is the starting point of our induction.

Now suppose $f(\alpha)f(\overline\alpha)=f(\beta)f(\overline\beta)$ holds for arbitrary $\alpha,\beta\in\su(f)$ satisfying $\dis{\alpha}{\beta}< 2t$. For $\alpha,\beta\in\su(f)$ satisfying $\dis{\alpha}{\beta}= 2t$, without loss of generality assume $\alpha=01\gamma,\beta=10\delta$. If $\alpha'=10\gamma\in \su(f)$, by induction we have $f(\alpha)f(\overline\alpha)=f(\alpha')f(\overline{\alpha'})=f(\beta)f(\overline\beta)$. Similarly if $\beta'=01\delta\in \su(f)$, by induction we have $f(\alpha)f(\overline\alpha)=f(\beta')f(\overline{\beta'})=f(\beta)f(\overline\beta)$. Now we assume $\alpha',\beta'\notin \su(f)$. By adding a self-loop on $x_1$ and $x_2$ we realize the signature $g$ of arity $2d-2$, and consequently $f(\alpha)f(\overline\alpha)=g(\gamma)g(\overline{\gamma})=g(\delta)g(\overline{\delta})=f(\beta)f(\overline\beta)$.
\end{proof}
Notice that, the condition "$k\ge 3$" is not necessary for the first statement "$f(\alpha)^k=f(\overline{\alpha})^k$" in the proof of Lemma \ref{complementary equal k}. Consequently, we may generalize the first statement as the following lemma for future reference.
\begin{lemma}\label{complementary equal-1}
    Assume $f$ is a signature with even arity and $k\geq1$ is an integer. If $B(f)=\{\bi{1}{\omega}\mid\omega^k=1\}$, then for arbitrary $\alpha\in\su(f)$, $f(\alpha)^k=f(\overline{\alpha})^k$. 
\end{lemma}
Now we are ready to prove that $f$ satisfies ARS ignoring a constant.
\begin{lemma}   \label{generating k}
    Assume $f$ an $\eo$ signatures and $k\geq3$ is an integer. If $B(f)=\{\bi{1}{\omega}\mid\omega^k=1\}$, then $f$ satisfies ARS ignoring a constant.
\end{lemma}
\begin{proof}
    By Lemma \ref{complementary equal k}, we may assume that for each $\alpha\in\su(f)$, $f(\alpha)f(\overline\alpha)=\lambda^2,\lambda\neq 0$. Furthermore, $f^k(\alpha)=f^k(\overline\alpha)$. Then $f(\alpha)/\lambda$ and $f(\overline\alpha)/\lambda$ are conjugate complex numbers for each $\alpha\in\su(f)$. So $f=\lambda \cdot g$ and $g$ satisfies ARS.
\end{proof}

\subsubsection{Generating no more than 2 signatures}\label{subsection:generating 2 signature}
Next we consider the case that $\bigcup_{f\in\mathcal{F}}B(f)=\{\neq_2,\bi{1}{-1}\}$. It is necessary to distinguish this case with the case that $B(f)=\{\neq_2\}$, so we commence with the characterization of the latter case. If for all $\alpha\in\su(f)$, $f(\alpha)=f(\overline{\alpha})$, then we say $f$ is \textit{domain-symmetric}, or simply \textit{D-sym}. By Lemma \ref{complementary equal-1} we have the following lemma. 
\begin{lemma}\label{complementary equal}
     Suppose $f$ is an $\eo$ signature of arity $2d$, and $B(f)=\{\neq_2\}$, then $f$ is D-sym.
\end{lemma}

Now we can further explain the claim we raised in Remark \ref{remark:decidableexistence}. It's noticeable that when all signatures in $\mathcal{F}$ are D-sym, all $\mathcal{F}$-gates are also D-sym, since the tensor operation and adding self-loop operation maintain this property. Therefore, if each signature in $\mathcal{F}$ can only realize $\neq_2$ in the generating process, we can equivalently say that it is D-sym, and consequently all binary $\mathcal{F}$-gates are of the form $\lambda\neq_2$. Deciding whether a signature is D-sym can obviously be done in constant time, thus the statement in Remark \ref{remark:decidableexistence} is true.

We then state a stronger property for signatures that can generate $\bi{1}{-1}$ through the generating process. Suppose $f$ is an $\eo$ signature. If for all $\alpha\in\su(f)$, $f(\alpha)=-f(\overline{\alpha})$, then we say $f$ is \textit{domain-anti-symmetric}, or simply \textit{DA-sym}.

\begin{lemma}\label{-1symmetric}
    Suppose $f$ is an $\eo$ signature of arity $2d$, and $B(f)=\{\neq_2,\bi{1}{-1}\}$, then $f$ is DA-sym.
\end{lemma}
\begin{proof}
    The proof is similar to that of Lemma \ref{complementary equal k}. By Lemma \ref{complementary equal-1}, for every $\alpha\in\su(f)$, $f(\alpha)/f(\overline{\alpha})=\pm 1$. We use induction to prove the following statement: For arbitrary $\alpha,\beta\in\su(f)$, if $B(f)\subseteq\{\neq_2,\bi{1}{-1}\}$, and $\neq_2,\bi{1}{-1}$ are available, then $f(\alpha)/f(\overline{\alpha})=f(\beta)/f(\overline{\beta})$. 
    
    When $2d=2$ the statement is trivially correct. Assume the statement is true for signatures of arity less than $2d$, we then consider $f$ of arity $2d$.

   We now introduce another induction on $\dis{\alpha}{\beta}$. We first consider $\alpha=01\delta,\beta=10\delta\in\su(f)$. According to Lemma \ref{complementary equal-1}, after adding a self-loop on $x_1$ and $x_2$, we have $(f(\alpha)+f(\beta))^2=(f(\overline{\alpha})+f(\overline{\beta}))^2$ and $(f(\alpha)-f(\beta))^2=(f(\overline{\alpha})-f(\overline{\beta}))^2$. Solving this we have $f(\alpha)/f(\overline{\alpha})=f(\beta)/f(\overline{\beta})$.

Now suppose $f(\alpha)/f(\overline{\alpha})=f(\beta)/f(\overline{\beta})$ holds for arbitrary $\alpha,\beta\in\su(f)$ satisfying $\dis{\alpha}{\beta}< 2t$. For $\alpha,\beta\in\su(f)$ satisfying $\dis{\alpha}{\beta}= 2t$, without loss of generality assume $\alpha=01\gamma,\beta=10\delta$. If $\alpha'=10\gamma\in \su(f)$, by induction we have $f(\alpha)/f(\overline\alpha)=f(\alpha')/f(\overline{\alpha'})=f(\beta)/f(\overline\beta)$. Similarly if $\beta'=01\delta\in \su(f)$, by induction we have $f(\alpha)/f(\overline\alpha)=f(\beta')/f(\overline{\beta'})=f(\beta)/f(\overline\beta)$. Now we assume $\alpha',\beta'\notin \su(f)$. By adding a self-loop on $x_1$ and $x_2$ we realize the signature $g$ of arity $2d-2$, and consequently $f(\alpha)/f(\overline\alpha)=g(\gamma)/g(\overline{\gamma})=g(\delta)/g(\overline{\delta})=f(\beta)/f(\overline\beta)$ since $B(g)\subseteq\{\neq_2,\bi{1}{-1}\}$.


    As a result, there exists a constant $\lambda
    \in \{1,-1\}$, such that for each $\alpha\in\su(f)$, $f(\alpha)/f(\overline{\alpha})=\lambda$. If $\lambda=1$, then $f$ is D-sym, and consequently  $B(f)=\{\neq_2\}$, a contradiction.
    Thus we have that for all $\alpha\in\su(f)$, $f(\alpha)=-f(\overline{\alpha})$.
\end{proof}
Now we specify all arity $4$ signatures that satisfy $B(f)\subseteq\{\neq_2,\bi{1}{-1}\}$.
\begin{lemma}\label{arity40110}
    Let $f\in\mathcal{F}$ be an $\eo$ signature of arity $4$. If $B(f)=\{\neq_2\}$, then ignoring a constant, $f=(\neq_2)^{\otimes2}$, or $f=({\bi{1}{-1}})^{\otimes2}$ , or $\ceoc(\mathcal{F})\leq_T\#\eo(\mathcal{F})$.
\end{lemma}
\begin{proof}
    By Lemma \ref{complementary equal}, $f$ is D-sym. And by Theorem \ref{thm:sixvertexdichotomyMform}, we can assume that $|\su(f)|=2$ or $4$, since otherwise $\#\eo(f)$ is \#P-hard. By Lemma \ref{complementary equal}, for each $\alpha\in\su(f)$, $f(\overline{\alpha})=\overline{f(\alpha)}$.

    If $|\su(f)|=2$, then $f$ can be written as $\ouhe$. By Lemma \ref{ouhetoeoc}, $\ceoc(f)\leq_T\#\eo(f)$.

    If $|\su(f)|=4$, then $f$ is of the form    
    $\left( \begin{array}{cccc}
0 & 0 & 0 & 0\\
0 & 1 & a & 0\\
0 & a & 1 & 0\\
0 & 0 & 0 & 0
\end{array} 
\right)$. If $f\in\mathscr{P}$, then it should be the tensor product of two same binary signatures, therefore $a^2=1$. If $f\in\mathscr{A}$, by the definition of $\mathscr{A}$ we have $a^2=\pm1$. With these two facts, we have that $a=\pm1,\pm\mathfrak{i}$, since otherwise $\#\eo(f)$ is \#P-hard by Theorem \ref{thm:sixvertexdichotomyMform}.
If $a=\pm\mathfrak{i}$, by $\left(\begin{array}{cccc}
0 & 0 & 0 & 0\\
0 & 1 & \pm\mathfrak{i} & 0\\
0 & \pm\mathfrak{i} & 1 & 0\\
0 & 0 & 0 & 0
\end{array} \right)\left(\begin{array}{cccc}
0 & 0 & 0 & 1\\
0 & 0 & 1 & 0\\
0 & 1 & 0 & 0\\
1 & 0 & 0 & 0
\end{array} \right)\left(\begin{array}{cccc}
0 & 0 & 0 & 0\\
0 & 1 & \pm\mathfrak{i} & 0\\
0 & \pm\mathfrak{i} & 1 & 0\\
0 & 0 & 0 & 0
\end{array} \right)=\left(\begin{array}{cccc}
0 & 0 & 0 & 0\\
0 & \pm2\mathfrak{i} & 0 & 0\\
0 & 0 & \pm2\mathfrak{i} & 0\\
0 & 0 & 0 & 0
\end{array} \right)$ we realize $\ouhe$, and consequently by Lemma \ref{ouhetoeoc} $\ceoc(\mathcal{F})\leq_T\#\eo(\mathcal{F})$.
If $a=1$, then $f=(\neq_2^{\otimes2})$ ignoring a constant.
If $a=-1$, then $f=(\bi{1}{-1})^{\otimes2}$  ignoring a constant. 
\end{proof}
\begin{lemma}\label{arity401-10}
    Assume $f\in \mathcal{F}$ is an $\eo$ signature of arity $4$, and $B(f)=\{\neq_2,\bi{1}{-1}\}$, then $f=(\neq_2)\otimes(\bi{1}{-1})$ ignoring a constant, or $\ceoc(\mathcal{F})\leq_T\#\eo(\mathcal{F})$.
\end{lemma} 
\begin{proof}
     As $B(f)=\{\neq_2,\bi{1}{-1}\}$, $f$ is not D-sym. We discuss several situations based on the size of $\su(f)$. With Lemma \ref{complementary equal-1} we can know that $|\su(f)|$ is even. We only consider the case that $\#\eo(f)$ is tractable, otherwise $\#\eo(f)$ is \#P-hard and the reduction trivially holds.

    If $|\su(f)|=2$, $f$ is of the form $\neq_4^{1,-1}$ by Lemma \ref{-1symmetric}. By Lemma \ref{lem:neq4ab obtain ouhe}, we may realize $\ouhe$. Therefore by Lemma \ref{ouhetoeoc} we have $\ceoc(\mathcal{F})\leq_T\#\eo(\mathcal{F})$.

If $|\su(f)|=4$, as $f$ is not D-sym, without loss of generality suppose it is of the form
$\left( \begin{array}{cccc}
0 & 0 & 0 & 0\\
0 & 1 & a & 0\\
0 & -a & -1 & 0\\
0 & 0 & 0 & 0
\end{array} \right)$. If $f\in\mathscr{P}$, then it should be the tensor product of two binary signatures, therefore $-a^2=-1$. If $f\in\mathscr{A}$, by Definition \ref{defa} we have $a^2=\pm1$. With these two facts, we have that $a=\pm1,\pm\mathfrak{i}$, since otherwise $\#\eo(f)$ is \#P-hard by Theorem \ref{thm:sixvertexdichotomyMform}. 
If $a=\mathfrak i$, by $\left(\begin{array}{cccc}
0 & 0 & 0 & 0\\
0 & 1 & \mathfrak{i} & 0\\
0 & -\mathfrak{i} & -1 & 0\\
0 & 0 & 0 & 0
\end{array} \right)\left(\begin{array}{cccc}
0 & 0 & 0 & 1\\
0 & 0 & 1 & 0\\
0 & 1 & 0 & 0\\
1 & 0 & 0 & 0
\end{array} \right)\left(
\begin{array}{cccc}
0 & 0 & 0 & 0\\
0 & 1 & \mathfrak{i} & 0\\
0 & -\mathfrak{i} & -1 & 0\\
0 & 0 & 0 & 0
\end{array} \right)=\left(\begin{array}{cccc}
0 & 0 & 0 & 0\\
0 & 0 & -2 & 0\\
0 & -2 & 0 & 0\\
0 & 0 & 0 & 0
\end{array} \right)$ we realize $\ouhe$, and consequently by Lemma \ref{ouhetoeoc} $\ceoc(\mathcal{F})\leq_T\#\eo(\mathcal{F})$. If $a=-\mathfrak{i}$, the analysis is similar by symmetry. 
If $a=\pm1$, then $f=(\neq_2)\otimes(\bi{1}{-1})$ ignoring a constant.
\end{proof}
Now we analyze the general cases.
\begin{lemma}\label{lem:generate 2 signature classification}
    Assume $f\in\mathcal{F}$ is an $\eo$ signature of arity $2d$. If $B(f)\subseteq\{\neq_2,\bi{1}{-1}\}$,  then $f$ satisfies ARS ignoring a constant, or $\ceoc(\mathcal{F})\leq_T\#\eo(\mathcal{F}\cup\{\neq_2,\bi{1}{-1}\})$.
\end{lemma}
\begin{proof}
    For convenience, when $\alpha$ and $\beta$ satisfy that $f(\alpha)f(\beta)\neq0$ and $f(\alpha)/f(\beta)\notin\mathbb{R}$, we denote by $\alpha\sim\beta$. On the other hands we use $\alpha\nsim\beta$ to denote the case that either $f(\alpha)f(\beta)=0$ or $f(\alpha)/f(\beta)\in\mathbb{R}$. We then use induction to prove that if $B(f)\subseteq\{\neq_2,\bi{1}{-1}\}$, then for arbitrary $\alpha,\beta\in\su(f)$, $\alpha\nsim\beta$, or otherwise $\ceoc(\mathcal{F})\leq_T\#\eo(\mathcal{F}\cup\{\neq_2,\bi{1}{-1}\})$.

    First, we prove that if there exists $\alpha,\beta\in\su(f)$ such that $\alpha\sim\beta$, then $\ceoc(\mathcal{F})\leq_T\#\eo(\mathcal{F}\cup\{\neq_2,\bi{1}{-1}\})$.  By Lemma \ref{arity40110} and \ref{arity401-10} we know that for all signatures of arity $2$ and $4$ the statement is true. Assume the statement holds for all signatures of arity less than $2d$, where $d\ge 3$. Now we consider $f$ of arity $2d$. 
    
     Suppose $\ceoc(\mathcal{F})\leq_T\#\eo(\mathcal{F}\cup\{\neq_2,\bi{1}{-1}\})$ does not holds. We prove the statement by introducing another induction on $\dis{\alpha}{\beta}$. We begin with the case that $\alpha,\beta\in\su(f)$, $\alpha\sim\beta$ and $\dis{\alpha}{\beta}=2$.
    Without loss of generality, we can assume that $\alpha=001011\gamma$, $\beta=010011\gamma$, $f(\alpha)=a$, $f(\beta)=b$, and $\alpha\sim\beta$. Also we can assume that $f(001101\gamma)=c$, $f(010101\gamma)=d$, $f(001110\gamma)=e$, $f(010110\gamma)=g$. These values actually forms a submatrix of $M(f)_{x_1x_2x_3,x_4x_5x_6\dots x_{2d}}$, presented as follows.  
        \begin{table}[h]
         \begin{center}
\begin{tabular}{l|lll}
      & $011\gamma$ & $101\gamma$ & $110\gamma$ \\ \hline
$001$ & $a$& $c$& $e$\\
$010$ & $b$& $d$& $g$\end{tabular}
\end{center}
\end{table}

     Because $a/b\notin\mathbb{R}$, we know that $a$, $b$ form a basis of the complex plane, which is a vector space defined on $\mathbb{R}$. As a result we can let $c=p_1a+p_2b$, $d=q_1a+q_2b$, where $p_1,p_2,q_1,q_2\in\mathbb{R}$.     
    By adding a self-loop between $x_2$ and $x_3$ respectively using $\neq_2$ and $\bi{1}{-1}$, according to the first induction hypothesis, we have $a+b\nsim c+d, a-b\nsim c-d$. As $a\sim b$, $(a+b)(a-b)\neq 0$, and consequently $p_1+q_1=p_2+q_2$, $p_1-q_1=q_2-p_2$. So $p_1=q_2$ and $q_1=p_2$ by solving the equations. Then we can suppose that $c=ra+tb$, $d=ta+rb$, where $r,t\in\mathbb{R}$. By adding a self-loop between $x_4$ and $x_5$ respectively using $\neq_2$ and $\bi{1}{-1}$, again with the first induction hypothesis we have that $(r+1)a+tb$ is collinear with $ta+(r+1)b$ and $(r-1)a+tb$ is collinear with $ta+(r-1)b$. Therefore $(r+1)^2=t^2=(r-1)^2$, furthermore $r=0$ and $t=\pm1$. So $c= \lambda b$, $d=\lambda  a$ for some $\lambda=\pm1$. Similarly $e=\eta b$ and $g=\eta a$ for some $\eta=\pm1$. We add a self-loop using  $\bi{1}{\lambda\eta}$ between $x_5$ and $x_6$, then we have $2a$ is collinear with $2b$, which is a contradiction. 

    Now assume that if $\ceoc(\mathcal{F})\leq_T\#\eo(\mathcal{F}\cup\{\neq_2,\bi{1}{-1}\})$  does not hold, then for arbitrary $\alpha,\beta\in\su(f)$ satisfying $\dis{\alpha}{\beta}<2t, t\ge 2$, we have $\alpha\nsim\beta$ . Now consider the case when $\alpha,\beta\in\su(f)$, $\alpha\sim\beta$ and $\dis{\alpha}{\beta}=2t$, while $\ceoc(\mathcal{F})\leq_T\#\eo(\mathcal{F}\cup\{\neq_2,\bi{1}{-1}\})$ does not hold. 
    Without loss of generality, we can assume that $\alpha=01\gamma$, $\beta=10\delta$, where $\gamma\neq\delta$. Let $\alpha'=10\gamma$, $\beta'=01\delta$. If $f(\alpha')\neq 0$, then either $f(\alpha')\sim f(\alpha)$ or $f(\alpha')\sim f(\beta)$ holds, a contradiction by the second induction hypothesis. The analysis for $f(\beta')$ is similar. Consequently, we may assume $f(\alpha')=f(\beta')=0$. By adding a self-loop on $x_1$ and $x_2$, we realize the signature $g$ of arity $2d-2$. Then we have $f(\alpha)=g(\gamma)$ ,$g(\delta)=f(\beta)$ and $\gamma\sim\delta$ in $g$. Then by the first induction hypothesis it is a contradiction.

    

    Consequently, we have proved that for a signature $f\in \mathcal{F}$ satisfying $B(f)\subseteq\{\neq_2,\bi{1}{-1}\}$, either $\ceoc(\mathcal{F})\leq_T\#\eo(\mathcal{F}\cup\{\neq_2,\bi{1}{-1}\})$, or for arbitrary $\alpha,\beta\in\su(f)$, $\alpha\nsim\beta$. If the latter statement holds, after normalization by a constant, $f$ is a real signature. If $B(f)=\{\neq_2,\bi{1}{-1}\}$, $f$ is DA-sym due to Lemma \ref{-1symmetric}, then $f=\mathfrak{i}\cdot g$ and $g$ satisfies ARS. If $B(f)=\{\neq_2\}$, $f$ is D-sym due to Lemma \ref{complementary equal}, then $f$ satisfies ARS. In summary, the lemma is proved.
\end{proof}


\begin{lemma}
    Assume $f\in\mathcal{F}$ is an $\eo$ signature of arity $2d$. If $B(f)=\{\neq_2\}$, then $f$ satisfies ARS ignoring a constant, or $\ceoc(\mathcal{F})\leq_T\#\eo(\mathcal{F})$.
    \label{lem:generate1}
\end{lemma}
\begin{proof}
Let $C(f)$ be the set of all non-trivial quaternary signatures, realized by adding $d-2$ self-loops on arbitrary $2d-4$ variables of $f$ by $\neq_2$. By Lemma \ref{arity40110}, $C(f)\subseteq\{(\neq_2)^{\otimes2},({\bi{1}{-1}})^{\otimes2}\}$, or $\ceoc(\mathcal{F})\leq_T\#\eo(\mathcal{F})$.

\textbf{Case 1:} Suppose $({\bi{1}{-1}})^{\otimes2}\in C(f)$. By Lemma \ref{lem:linwang}, we have $\ceo(\mathcal{F}\cup\{\bi{1}{-1}\})\leq_T\#\eo(\mathcal{F})$. Then we are done by Lemma  \ref{lem:generate 2 signature classification}. 
    
\textbf{Case 2:} $C(f)=\emptyset$. If $f\equiv 0$, it satisfies ARS. Otherwise by Lemma \ref{lem:self-loop to nontrivial eo} we may realize a non-trivial quaternary signature by adding $d-2$ self-loops, a contradiction.
    

\textbf{Case 3:} $C(f)=\{(\neq_2)^{\otimes2}\}$.
    Let $f_{\text{Re}}$ and $f_{\text{Im}}$ be the signatures satisfying that, for each $\alpha$ of length $2d$, $f_{\text{Re}}(\alpha)=\text{Re}(f(\alpha))$, $f_{\text{Im}}(\alpha)=\text{Im}(f(\alpha))$. Let $I$ be a specific instance formed by adding $d$ self-loops on $f$, and $g$ be a quaternary signature formed by adding $d-2$ self-loops on $f$. 
   Let $I_{\text{Re}}$ be the instance obtained by substituting $f$ with $f_{\text{Re}}$ in $I$, $I_{\text{Im}}$ be the instance obtained by substituting $f$ with $f_{\text{Im}}$ in $I$. We define $g_{\text{Re}},g_{\text{Im}}$ similarly.

    It is noticeable that given a way of adding $d$ self-loops, the value of this single-vertex tensor network is a linear combination of all possible values of $f$, where the coefficients are in $\{0,1\}$.  Therefore $\text{Z}(I)=\text{Z}(I_{\text{\text{Re}}})+\mathfrak{i}\text{Z}(I_{\text{Im}})$. Similarly we have $g=g_{\text{\text{Re}}}+\mathfrak{i}g_{\text{Im}}$. 

We now introduce the reorganizing method. First we can notice that, if we add  two self-loops on the signature $\neq_2^{\otimes2}$, the value of this tensor network is either $2$ or $4$. In other words, changing ways of adding self-loops on $\lambda\neq_2^{\otimes2},\lambda\in\mathbb{C}$ will cause the partition function to double, to be divided into half or to remain unchanged. 

 Let $\lambda=\text{Z}(I)$, and suppose $\lambda=a+b\mathfrak{i}$ where $a,b$ are real numbers. By cutting two connected edges in $I$, we obtain a gadget that realize $g=\eta\neq_2^{\otimes2}$ for some $\eta=c+d\mathfrak{i}\in\mathbb{C}$. By reorganizing the connection of the four dangling edges, the obtained instance $I'$ satisfies $\text{Z}(I')=2^t\text{Z}(I)$ for some $t\in\{0,\pm1\}$. The similar analysis also holds for $I'_{\text{Re}}$ and  $I'_{\text{Im}}$ defined similarly, and in particular, $\text{Z}(I'_{\text{Re}})=2^t\text{Z}(I_{\text{Re}}),\text{Z}(I'_{\text{Im}})=2^t\text{Z}(I_{\text{Im}})$ holds for the same $t$ as that in $\text{Z}(I')=2^t\text{Z}(I)$. This is because  $g=g_{\text{\text{Re}}}+\mathfrak{i}g_{\text{Im}}$ and consequently adding two self-loops on $g,g_{\text{\text{Re}}}$ and $g_{\text{Im}}$ in the same way would introduce the same constant factor.
    
     If $a\neq0$, by applying the reorganizing method above successively, we know that for all $I'$ obtained by adding $d$ self-loops on $f$, by substituting $f$ with $f_{\text{Re}}$ the obtained instance $I'_{\text{Re}}$ always satisfies $\text{Z}(I')=\frac{\lambda}{a}\text{Z}(I'_{\text{Re}})$. Noticing that given a way of adding $d$ self-loops, when we write all values of $f$ as a column vector $\mathbf{f}$, the way of connecting edge is equivalent to a row vector $\mathbf{v}$ such that $\mathbf{vf}$ is the partition function. Therefore we have a system of linear equations, whose coefficient matrix is an arrangement of row vectors corresponding to all ways of adding $d$ self-loops. Suppose the coefficient matrix is $M$, then we have $M\mathbf{f}=\frac{\lambda}{a}M\mathbf{f_{\text{Re}}}$. 

    Here we prove an interesting fact that given any D-sym $f$ of arity $2d$, if $M\mathbf{h}=\mathbf{0}$, then $\mathbf{h}=\mathbf0$. Suppose $\mathbf{h}\neq\mathbf0$, which means $f$ is non-trivial. By Lemma \ref{lem:self-loop to nontrivial eo}, by adding $d-1$ self-loops properly we may realize a non-trivial binary signature $g$. Since $f$ is D-sym, $g$ is a multiple of $\neq_2$ and by adding one more self-loop on $g$ the resulting partition function is not equal to $0$. Consequently $M\mathbf{h}=\mathbf{0}$, and we are done.

    By the analysis above, we have that $\mathbf{f}=\frac{\lambda}{a}\mathbf{f_{\text{Re}}}$ and $f=\frac{\lambda}{a}f_{\text{Re}}$. This means that $f$ is a multiple of a real signature $f_{\text{Re}}$. If $b\neq0$, the analysis is similar and we have $f$ is a multiple of a real signature $f_{\text{Im}}$. As $f$ is D-sym, both $f_{\text{Re}}$ and $f_{\text{Im}}$ are D-sym as well, and consequently they satisfy ARS. Hence $f$ also satisfies ARS ignoring a constant.

    The only case left is when $a=b=0$, indicates that $\lambda=0$. By the reorganizing method we know that for all $I'$ obtained by adding $d$ self-loops on $f$, there exists $t'\in \mathbb{Z}$ such that  $\text{Z}(I')=2^{t'}\lambda=0$. This contradicts to  $C(f)=\{(\neq_2)^{\otimes2}\}$.
\end{proof}

\subsection{Put things together}
In this section, we prove Theorem \ref{thm:eoc to eo}.
\begin{proof}[Proof of Theorem \ref{thm:eoc to eo}]
    If there exists $f\in \mathcal{F}$, such that $\neq_2^{1,x}\in B(f)$ where $x$ is not a root, we are done by analysis in Section \ref{subsec:EOC part: nonroot-interpolation}.  Otherwise if there exists $f\in \mathcal{F}$, such that $B(f)$ is infinite, we are done by the analysis in Section \ref{subsection: infinitely many roots}.

    Consequently we may assume for each $f\in \mathcal{F}$, there exists an integer $k\geq1$ such that $B(f)=\{\bi{1}{\omega}\mid\omega^k=1\}$. By Lemma \ref{generating k}, \ref{lem:generate 2 signature classification} and \ref{lem:generate1}, either $f$ satisfies ARS  ignoring a constant or  $\ceoc(\mathcal{F})\leq_T\#\eo(\mathcal{F})$. If $\ceoc(\mathcal{F})\leq_T\#\eo(\mathcal{F})$ we are done by the second statement. Otherwise all signatures in $\mathcal{F}$ satisfy ARS ignoring a constant and we are done by the first statement.
\end{proof}

\section{Dichotomy for $\#\eo^c$} \label{section: dichotomy for eoc}

In this section, we present the $\pnp$ versus \#P dichotomy for $\#\eo^c$ problem and prove Theorem \ref{thm:dicoceoc}.
We begin with the analysis of the support of each $f\in \mathcal{F}$.
By definition, $\mathcal{F}$ always falls into one of the following cases:
\begin{enumerate}
    \item There exists a $\exists3\nrightarrow$ signature $f\in \mathcal{F}$;
    \item There exist a $\exists3\uparrow$ signature $f\in \mathcal{F}$ and a $\exists3\downarrow$ signature $g\in \mathcal{F}$;
    \item All signatures in $\mathcal{F}$ are $\forall 3 \thup$ signatures:
    \begin{enumerate}
    \item For each $f\in \mathcal{F}$, $f$ is a $\allup$ signature;
    \item For each $f\in \mathcal{F}$, $f\in \eom[\mathscr{A}]$;
    \item For each $f\in \mathcal{F}$, $f\in \eom[\mathscr{P}]$;
    \item Otherwise;
\end{enumerate}
    \item All signatures in $\mathcal{F}$ are $\forall 3 \thdown$ signatures.
\end{enumerate}
 
We deal with these cases in the following sections. We omit the analysis of Case 4 as it is similar to that of Case 3. To summarize, $\ceoc(\mathcal{F})$ is \#P-hard if $\mathcal{F}$ falls into Case 1, 2, 3d, 4d, while it has a $\pnp$ algorithm if $\mathcal{F}$ falls into Case 3b, 3c, 4b, 4c. If $\mathcal{F}$ falls into Case 3a or 4a, $\ceoc(\mathcal{F})$ is \#P-hard unless $\mathcal{F}\subseteq \eom[\mathscr{A}]$ or $\mathcal{F}\subseteq \eom[\mathscr{P}]$, in which cases it is P-time computable. It is noteworthy that with the appearance of $\Delta$, for each \eo\ signature $f$, now we are able to pin a variable of $f$ to 0 and another variable of $f$ to 1 by gadget construction.

\subsection{Case 1}
In this section, we prove the following lemma:
\begin{lemma}
    If $f$ is a $\exists3\nrightarrow$ signature, then $\#\eo^c(f)$ is \#P-hard.
    \label{lem:3midhard}
\end{lemma}

The proof of Lemma \ref{lem:3midhard} consists of 2 parts. The first part uses an inductive method to reduce the arity of $f$. The second part proves the \#P-hardness when $f$ is of small arity. We first presents an important method for reducing the arity in the following lemma.

\begin{lemma}
    Suppose $f$ is an $\eo$ signature of arity $2k$ and $\alpha=10\alpha',\beta=10\beta',\gamma=01\gamma'\in \su(f)$, where $k\geq3$. Then there exists an $\eo$ signature $g$ of arity $2k-2$ satisfying $\alpha',\beta',\gamma'\in \su(g)$, such that $\#\eo^c(g)\leq_T\#\eo^c(f)$. Furthermore, if $\delta'=\alpha'\oplus\beta'\oplus\gamma'\in \eoe$ and $10\delta',01\delta'\notin\su(f)$, then $\delta'\notin \su(g)$.
    \label{lem:3keep}
\end{lemma}
\begin{proof}
    If $10\gamma'\in \su(f)$, by pinning $x_1$ to 1 and $x_2$ to 0, the obtained signature $g$ satisfies that $\alpha',\beta',\gamma'\in \su(g)$. Else if both $01\alpha',01\beta'\in \su(f)$, by pinning $x_1$ to 0 and $x_2$ to 1, the obtained signature $g$ satisfies that $\alpha',\beta',\gamma'\in \su(g)$. In addition, if $10\gamma',01\alpha',01\beta'\notin \su(f)$, by adding a self-loop on $x_1$ and $x_2$, the obtained signature $g$ satisfies that $\alpha',\beta',\gamma'\in \su(g)$. In the cases above, we have $\#\eo^c(g)\leq_T\#\eo^c(f)$.

    Now we may assume $01\alpha',10\gamma'\notin \su(f)$, while $01\beta'\in\su(f)$. If $f(10\beta')+ f(01\beta')\neq 0$, by adding a self-loop to $x_1$ and $x_2$, the obtained signature $g$ satisfies $\alpha',\beta',\gamma'\in \su(g)$. Otherwise, $f(10\beta')+ f(01\beta')= 0$, and by pinning $x_3,...,x_{2k}$ to $\beta'$ we obtain the signature $\bi{1}{-1}$. By adding a self-loop by $\bi{1}{-1}$ on $x_1$ and $x_2$, the obtained signature $g$ again satisfies that $\alpha',\beta',\gamma'\in \su(g)$. The case that $01\beta',10\gamma'\notin\su(f)$ and $01\alpha'\in\su(f)$ is similar.

    In each case, since $g$ is obtained from $f$ by adding a self-loop on $x_1$ and $x_2$ by some generalized binary disequality, $\delta'\notin \su(g)$ if $10\delta',01\delta'\in \eoe-\su(f)$.
\end{proof}
We then summarize the second part as the following lemmas.
\begin{lemma}
    Suppose $f$ is an $\eo$ signature of arity 4. If $\alpha,\beta,\gamma\in \su(f)$ and $\delta=\alpha\oplus\beta\oplus\gamma\in \eoe-\su(f)$, then $\#\eo^c(f)$ is \#P-hard. 
    \label{lem:4nomid}
\end{lemma}
\begin{proof}
    The proof follows from Theorem \ref{thm:arity4setdichotomy}.
\end{proof}
\begin{lemma}
    Suppose $f$ is an $\eo$ signature of arity 6. If $\alpha=101001,\beta=100110,\gamma=011010\in \su(f)$ and $\delta=010101\notin \su(f)$, then $\#\eo^c(f)$ is \#P-hard. 
    \label{lem:6nomid}
\end{lemma}
\begin{proof}
    Suppose otherwise. If $\epsilon=100101\notin \su(f)$, by adding a self-loop on $x_1$ and $x_2$ by some generalized binary disequality as in Lemma \ref{lem:3keep}, we can obtain a signature $g$ satisfying $1001,0110,1010\in \su(g)$ while $0101\notin \su(g)$. By Lemma \ref{lem:4nomid}, $g$ is \#P-hard, which is a contradiction. Thus $\epsilon=100101\in \su(f)$. Similarly, by adding a self-loop to $x_3$ and $x_4$, $\zeta=011001\in \su(f)$. By pinning $x_5$ to 0 and $x_6$ to 1, we obtain a signature $h$ satisfying $1010,1001,0110\in \su(h)$ and $0101\notin \su(h)$. Again by Lemma \ref{lem:4nomid}, $h$ is \#P-hard, which is a contradiction.
\end{proof}

Before the formal proof of Lemma \ref{lem:3midhard}, we introduce another lemma for preparation. 

\begin{lemma} \label{lem:4 string all eo column property}
    Suppose $\alpha$, $\beta$, $\gamma$, $\delta$ are 01-strings with the same even length. $\alpha,\beta,\gamma,\delta\in\eoe$ and $\delta=\alpha\oplus\beta\oplus\gamma$. For $a,b,c\in\{0,1\}$, let $S_{abc}=\{i|\alpha_i=a,\beta_i=b,\gamma_i=c\}$ and $s_{abc}=|S_{abc}|$. Then $s_{abc}=s_{1-a,1-b,1-c}$ for arbitrary $a,b,c\in\{0,1\}$.
\end{lemma} 
\begin{proof}
    There are 8 possible values of $(a,b,c)$, that is $000,001,010,011,100,101,110,111$. Since $\alpha,\beta,\gamma,\delta\in\eoe$, we have: 
    $$\begin{cases}
s_{000}+s_{001}+s_{010}+s_{011}=s_{100}+s_{101}+s_{110}+s_{111}\\
s_{000}+s_{001}+s_{100}+s_{101}=s_{010}+s_{011}+s_{110}+s_{111}\\
s_{000}+s_{010}+s_{100}+s_{110}=s_{001}+s_{011}+s_{101}+s_{111}\\
s_{000}+s_{011}+s_{101}+s_{110}=s_{001}+s_{010}+s_{100}+s_{111}
    \end{cases}$$
    Solving this linear system, we have $s_{abc}=s_{1-a,1-b,1-c}$ for arbitrary $a,b,c\in\{0,1\}$.
\end{proof}

Now we are ready to prove Lemma \ref{lem:3midhard}.
\begin{proof}
    By Lemma \ref{lem:4nomid} we have that the conclusion is correct when $f$ is of arity 4. Now suppose $f$ is of arity $2k$, where $k\geq3$, $\alpha,\beta,\gamma\in \su(f)$ and $\delta=\alpha\oplus\beta\oplus\gamma\in \eoe-\su(f)$. We first characterize the form of these strings.

For $a,b,c\in\{0,1\}$, let $S_{abc}=\{i|\alpha_i=a,\beta_i=b,\gamma_i=c\}$ and $s_{abc}=|S_{abc}|$. Using the condition that $\alpha,\beta,\gamma,\delta\in \eoe$ and Lemma \ref{lem:4 string all eo column property}, we have $s_{abc}=s_{1-a,1-b,1-c}$ for arbitrary $a,b,c\in\{0,1\}$.

In the following analysis, we present operations that reduce the arity of $f$ by 2 and keep $f$ as a $\exists3\nrightarrow$ signature.
If $s_{000}\ge1$, then $s_{111}\ge1$. By renaming the variables properly, we can write $\alpha=01\alpha',\beta=01\beta',\gamma=01\gamma',\delta=01\delta'$. By pinning $x_1$ to 0 and $x_2$ to 1, we obtain a signature $g$ satisfying $\alpha',\beta',\gamma'\in \su(g)$ and $\delta'\in \eoe-\su(g)$, which means $g$ is also a $\exists3\nrightarrow$ signature. We denote this operation as Operation 1.

If $s_{001}\ge2$, then $s_{110}\ge2$. By renaming the variables properly, we can write $\alpha=0101\alpha',\beta=0101\beta',\gamma=1010\gamma',\delta=1010\delta'$. Furthermore, let $\epsilon=0110\delta',\zeta=0110\gamma'$. 
\begin{itemize}
    \item If $\epsilon\notin \su(f)$, by adding a self-loop on $x_1$ and $x_2$ by some generalized binary disequality as in Lemma \ref{lem:3keep}, we can obtain a signature $g$ satisfying $01\alpha',01\beta',10\gamma'\in \su(g)$ while $10\delta'\notin \su(g)$.
    \item If $\epsilon=0110\delta'\in \su(f)$ and $\zeta=0110\gamma'\notin \su(f)$, by pinning $x_1$ to 0 and $x_2$ to 1, we obtain a signature $g$ satisfying $01\alpha',01\beta',10\delta'\in \su(g)$ and $10\gamma'\notin \su(g)$. 
    \item If $\epsilon,\zeta\in \su(f)$, by pinning $x_3$ to 1 and $x_4$ to 0, we obtain a signature $g$ satisfying $10\gamma',01\gamma',01\delta'\in \su(g)$ and $10\delta'\notin \su(g)$. 
\end{itemize}

In each situation, we can always obtain a $\exists3\nrightarrow$ signature $g$. We denote this operation as Operation 2. Similarly, we define Operation 3 (or 4) in the situation that $s_{010}\ge2$ (or $s_{100}\ge2$).

By applying Operation 1-4 successively, one can reduce the arity of $f$ until $s_{000}=s_{111}=0$ and $s_{001}=s_{110}\le 1,s_{010}=s_{101}\le 1,s_{100}=s_{011}\le 1$, maintaining $f$ as a $\exists3\nrightarrow$ signature. As $f$ is a $\exists3\nrightarrow$ signature, the arity of $f$ is not less than 4 and thus $s_{001}+s_{010}+s_{100}\ge 2$. If $s_{001}+s_{010}+s_{100}=2$, the \#P-hardness of $\#\eo^c(f)$ is given by Lemma \ref{lem:4nomid}. If $s_{001}+s_{010}+s_{100}=3$, the \#P-hardness of $\#\eo^c(f)$ is given by Lemma \ref{lem:6nomid}.
\end{proof}

\subsection{Case 2}
In this section, we show that Lemma \ref{lem:3midhard} also induces the \#P-hardness of Case 2.
\begin{lemma}
    If $f$ is a $\exists3\uparrow$ signature and $g$ is a $\exists3\downarrow$ signature, then $\#\eo^c(\{f,g\})$ is \#P-hard.
\end{lemma}
\begin{proof}
    Suppose $\alpha,\beta,\gamma\in \su(f)$, $\delta=\alpha\oplus\beta\oplus\gamma\in\eosg$ and $\epsilon,\zeta,\eta\in \su(g)$, $\theta=\epsilon\oplus\zeta\oplus\eta\in\eosl$. Let $a=\#_1(\delta)-\#_0(\delta)$ and $b=\#_0(\theta)-\#_1(\theta)$. Now we consider the signature $h=f^{\otimes b}\otimes g^{\otimes a}$. It can be verified by the definition of the tensor product operation that $\alpha^b \epsilon^a,\beta^b \zeta^a,\gamma^b \eta^a\in \su(h)$, while $\alpha^b \epsilon^a\oplus\beta^b \zeta^a\oplus\gamma^b \eta^a= \delta^b \theta^a\in \eoe-\su(h)$. By Lemma \ref{lem:3midhard}, $\#\eo^c(h)$ is \#P-hard, and consequently $\#\eo^c(\{f,g\})$ is \#P-hard as well.
\end{proof}
We remark that this lemma also works when $f=g$.

\subsection{Case 3a}\label{allupp}
By Theorem \ref{thm:puredichotomy}, and the fact that the definition of pure-up signatures is the same as that of $\allup$, its complexity is already classified. 

\subsection{Case 3b and 3c}\label{3upeasy}
In this section, we prove the following lemma.
\begin{lemma}
    Suppose all signatures in $\mathcal{F}$ are $\mitsuup$ (or $\mitsudown$) signatures. If $\mathcal{F}\subseteq \eom[\mathscr{A}]$ or $\mathcal{F}\subseteq\eom[\mathscr{P}]$, then $\ceo(\mathcal{F})$ can be computed in polynomial time, given a specific NP oracle related to $\mathcal{F}$.
    \label{lem:NPalg}
\end{lemma}

We focus on the case when all signatures in $\mathcal{F}$ are $\mitsuup$ signatures. The analysis for the $\mitsudown$ case is similar by exchanging the symbol of 0 and 1.
We commence with the decision version of counting weighted Eulerian orientations problems, which is denoted as the $\eo$ problem. 

\begin{definition}[$\eo$ problems]\cite{cai2020beyond}\label{def:decideeo}
A $\eo$ problem $\eo(\mathcal{F})$, parameterized by a set $\mathcal{F}$ of $\eo$ signatures, is defined as follows: Given an $\eo$-signature grid $\Omega(G, \pi)$ over $\mathcal{F}$, the output is whether there exists $\sigma \in \eo(G)$ such that 

$$
\prod_{v \in V} f_v(\sigma|_{E(v)})\neq0.
$$
\end{definition}

By Definition \ref{def:decideeo}, given a set of signatures $\mathcal{F}$, we can change all non-zero values of signatures in $\mathcal{F}$ to 1 without changing the output of $\eo(\mathcal{F})$. We also denote such an assignment $\sigma$ in Definition \ref{def:decideeo} as a solution.

We then introduce the concept of the effective string in $\eo$ problems, and show that deciding which string is effective in the instance can be computed by an NP oracle. These would provide us a useful method toward our algorithm for the counting version.
\begin{definition}
    Suppose $I$ is an instance of $\eo(\mathcal{F})$. For a vertex $v$ in $I$ assigned $f\in \mathcal{F}$, we say a string $\alpha_f\in\su(f)$ is effective if there exists a solution $\alpha$ to $I$ satisfying $\alpha|_{E(v)}=\alpha_f$.    
    The set of all effective strings is denoted as the effective support of $f$.
    \label{def:effective}
\end{definition}
We remark that in Definition \ref{def:effective}, the concept of effective strings and effective support is defined on the specific $f$ that assigned to $v$, instead of the element $f\in\mathcal{F}$. For example, for a signature $f\in\mathcal{F}$, suppose it is assigned to $u$ and $v$ in an instance $I$, which is then denoted by $f_u$ and $f_v$ respectively. Then the effective support of $f_u$ and $f_v$ can be totally different and we always analyze them separately in the following proof. Besides, whether a string is effective in an instance define the following problem.
\begin{definition}[Support identification problem]
    A support identification problem is parameterized by a set $\mathcal{F}$ of \eo\ signatures.
    
    \textbf{Input:} An instance $I$ of $\eo(\mathcal{F})$, a vertex assigned $f$ in $I$ and a string $\alpha_f\in\su(f)$.

    \textbf{Output:} Whether $\alpha_f$ is effective.
\end{definition}

The length of each string $\alpha=\alpha_f\alpha'$ is $n$, where $n$ is the number of variables in the instance. Consequently, the support identification problem is in NP, which means that  given $\mathcal{F}$, there exists a consistent NP oracle solving it.
We also remark that the computational complexity of the support identification problem is still open, even when $\mathcal{F}$ falls into Case 3 (or Case 4). That is, given $\mathcal{F}$, whether this problem can be computed in polynomial time or it is NP-complete remains unsettled. Dozens of attempts have been made, but the authors are still facing non-trivial difficulties.

It is notable that the concept of \eo\ problem shares common symbols with $\ceo$ in our definition, which enables us to extend the concept of effective strings to the counting version $\ceo$. For an instance $I$ in $\ceo$, a vertex assigned $f\in \mathcal{F}$ in $I$, we say a string $\alpha_f\in \su(f)$ is effective if there exists a $\alpha=\alpha_f\alpha'$ such that $\prod_{v\in V}f_v(\alpha|_{E(v)})\neq 0$. We also call $\alpha$ as a solution in this setting. Besides, it is worth noticing that whether a string in $\su(f)$ is effective is consistent in the setting of \eo\ and \ceo.

Using the concept of effective strings, we can modify signatures in a particular instance. The following lemma can be easily verified from Definition \ref{def:effective}.

\begin{lemma}
    Suppose $I$ is an instance of $\ceo(\mathcal{F})$, $v$ is a vertex in $I$ assigned $f\in\mathcal{F}$ and $\alpha_f\in\su(f)$ is not effective. Let $f'$ be the signature that only differs from $f$ at $\alpha_f$ and $f'(\alpha_f)=0$. Let $I'$ be the instance obtained by replacing $f$ by $f'$ at $v$. Then $\text{Z}(I')=\text{Z}(I)$. 
    
    Furthermore, by restricting $f$ to its effective support $\mathscr{S}$, we can modify $f$ to $f|_{\mathscr{S}}$ while keeping the partition function unchanged.
    \label{lem:suppmodi}
\end{lemma}

Now we focus on Case 3b and 3c.
In the following, we analyze the effective supports in a given instance $I$. 

\begin{lemma}
    Suppose an instance $I$ of $\#\eo(\mathcal{F})$ is given, where all signatures in $\mathcal{F}$ are $\mitsuup$. For a vertex assigned $f\in \mathcal{F}$ in $I$, if $\alpha_f,\beta_f,\gamma_f\in\su(f)$ and $\delta_f=\alpha_f\oplus\beta_f\oplus\gamma_f\in\eosg$, then at most two of $\alpha_f,\beta_f,\gamma_f$ can be effective.
    \label{lem:3upeff}
\end{lemma}

\begin{proof}
    Suppose otherwise. 
     We regard $I$ as an instance of the equivalent problem $\hol(\neq_2\mid \mathcal{F})$.  As all of $\alpha_f,\beta_f,\gamma_f$ are effective, there exists $\alpha=\alpha_f\alpha',\beta=\beta_f\beta',\gamma=\gamma_f\gamma'$ such that 
    each of them is a solution. Let $\delta=\delta_f\delta'$, where $\delta=\alpha\oplus\beta\oplus\gamma$. 

   Each variable in $I$ is incident to a signature in $\mathcal{F}$ on RHS and a $\neq_2$ on LHS, thus both sides cover each variable exactly once. For RHS, all signatures in $\mathcal{F}$ are $\mitsuup$ signatures, so $\delta'\in\eog$. Since $\delta_f\in\eosg$, we have $\delta\in\eosg$. However, by the restriction from LHS, we have $\delta\in \eoe$ since $\neq_2$ is an $\eom$ signature, which is a contradiction.
\end{proof}

\begin{lemma}
    Suppose an instance $I$ of $\#\eo(\mathcal{F})$ is given, where $\mathcal{F}$ contains only $\mitsuup$ signatures. For a signature $f\in \mathcal{F}$ in $I$, if $\alpha_f,\beta_f,\gamma_f\in\su(f)$ are effective, and $\delta_f=\alpha_f\oplus\beta_f\oplus\gamma_f\in\su(f)$, then $\delta_f$ is also effective.
    \label{lem:3mideff}
\end{lemma}
\begin{proof}
     We regard $I$ as an instance of the equivalent problem $\hol(\neq_2\mid \mathcal{F})$. Suppose $I$ has $n$ vertices $v_1,...,v_n$ assigned $f_1,...,f_n$ chosen from $\mathcal{F}$. Let $f=f_1$, and we use $\alpha_i$ to denote $\alpha|_{E(v_i)}$. As $\alpha_1,\beta_1,\gamma_1\in\su(f_1)$ are effective, there exist $\alpha=\alpha_1...\alpha_n,\beta=\beta_1...\beta_n,\gamma=\gamma_1...\gamma_n$ such that each of them is a solution.

    Let $\delta=\alpha\oplus\beta\oplus\gamma$. If $\delta\in\eosg$, there would exists a signature $f_i$ in $I$, such that $\alpha_i\oplus\beta_i\oplus\gamma_i\in\eosg$. Since $\alpha_i,\beta_i,\gamma_i$ are effective, by Lemma \ref{lem:3upeff} this would cause a contradiction. Therefore, $\delta\in\eoe$, and as all signatures in $\mathcal{F}$ are $\mitsuup$ signatures, $\delta_i\in\eoe$ and $\delta_i\in\su(f)$ for each $i$. This indicates that $\delta$ is also a solution, which implies $\delta_1$ is effective as well.
\end{proof}

We are now ready to present the algorithm for Lemma \ref{lem:NPalg}.
\begin{proof}
    The algorithm for Lemma \ref{lem:NPalg} is as follows.
    \begin{enumerate}
        \item For each string $\alpha_f\in\su(f)$ at some vertex assigned $f$ in the instance $I$: Decide whether $\alpha_f$ is effective. If not, replace $f$ with $f'$, where $f'$ only differs from $f$ at $\alpha_f$ and $f'(\alpha_f)=0$.
        \item Using the algorithm for $\mathscr{A}$ or that for $\mathscr{P}$ to compute the value of the modified instance $I'$. Output the value.
    \end{enumerate}
    By Lemma \ref{lem:suppmodi}, the operation in Step (1) do not change the partition function. After all modifications in Step (1), for each signature $f'$ in $I'$, if $\alpha,\beta,\gamma\in\su(f')$, by Lemma \ref{lem:3upeff} we have $\delta=\alpha\oplus\beta\oplus\gamma\in \eoe$. Furthermore, by Lemma \ref{lem:3mideff}, $\delta\in\su(f')$, which indicates that $\su(f')\subseteq\eoe$ is affine. By Lemma \ref{lem: eo+affine= eom}, 
    
    $\su(f')\subseteq\eo^P$ for some pairing $P$. Consequently, all signatures in $I'$ are in $\mathscr{A}$ (or all signatures in $I'$ are in $\mathscr{P}$) by Definition \ref{def:eoaeop}. Then the polynomial-time algorithm for $\mathscr{A}$ or that for $\mathscr{P}$ can be applied to $I'$.

    Now we analyze the time complexity of the algorithm. For Step (1), checking whether a string is effective need an NP oracle, which is determined given $\mathcal{F}$.  There are $O(n)$ strings need to be checked, where $n$ is the number of vertices in $I$. Step (2) can be computed in polynomial time. In summary, the algorithm is in $\pnp$.
\end{proof}

The algorithm is similar for the $\mitsudown$ cases. We remark that this algorithm also works for the 0-rebalancing cases by the following lema.
\begin{lemma}
    If $f$ is a 0-rebalancing signature, then $f$ is either $\exists3\nrightarrow$ or $\mitsuup$.
\end{lemma}
\begin{proof}
    We would prove the inverse negative proposition. Suppose $f$ is neither $\exists3\nrightarrow$ nor $\mitsuup$. Consequently, there exist $\alpha,\beta,\gamma\in\su(f)$ and $\delta=\alpha\oplus\beta\oplus\gamma\in \eosl$. Similar to Lemma \ref{lem:4 string all eo column property}, we consider $S_{abc}$ and $s_{abc}$, where $a,b,c\in\{0,1\}$. Since $\alpha,\beta,\gamma\in\eoe$ and $\delta\in\eosl$, we have 
    $$\begin{cases}
s_{000}+s_{001}+s_{010}+s_{011}=s_{100}+s_{101}+s_{110}+s_{111}\\
s_{000}+s_{001}+s_{100}+s_{101}=s_{010}+s_{011}+s_{110}+s_{111}\\
s_{000}+s_{010}+s_{100}+s_{110}=s_{001}+s_{011}+s_{101}+s_{111}\\
s_{000}+s_{011}+s_{101}+s_{110}>s_{001}+s_{010}+s_{100}+s_{111}
    \end{cases}.$$ Solving this we have $s_{111}<s_{000}$. We fix all variables assigned $000$ to $0$, then there would not be enough variables fixed to $1$ to appear. Therefore, $f$ cannot be 0-rebalancing.
\end{proof}

\subsection{Case 3d}
In this section, we prove the following lemma.
\begin{lemma}
    Suppose each signature in $\mathcal{F}$ is a $\mitsuup$ signature. If there exists a signature $f\in\mathcal{F}$ that is not a $\allup$ signature, and there exists a signature $g\in\mathcal{F}-\eo^{\mathscr{A}}$ and a signature $h\in\mathcal{F}-\eo^{\mathscr{P}}$, then $\ceoc(\mathcal{F})$ is \#P-hard.
    \label{3uphard}
\end{lemma}
We first show that, given the $\ouhe$ signature, we can prove the \#P-hardness. 

\begin{lemma}
    Suppose each signature in $\mathcal{F}$ is a $\mitsuup$ signature. If there exist a signature $g\in\mathcal{F}-\eo^{\mathscr{A}}$ and a signature $h\in\mathcal{F}-\eo^{\mathscr{P}}$, then $\ceo(\mathcal{F}\cup\{\ouhe\})$ is \#P-hard.
    \label{lem:ouhehard}
\end{lemma}
\begin{proof}
    For clarity, we let $\su(\ouhe)=\{0101,1010\}$. Suppose $g$ is of arity $2k$. As $g\notin \eo^{\mathscr{A}}$, by renaming the variables of $g$, there exists a pairing $P=\{(x_1,x_2),...,(x_{2k-1},x_{2k})\}$ such that $g|_{\eom[P]}\notin \mathscr{A}$. We can construct a gadget for $g|_{\eom[P]}$ by connecting the first and the second variables of a $\ouhe$ signature to each pair $x_{2i-1},x_{2i}$, $1\leq i\leq k$. Similarly we can construct $h|_{\eom[P']}$ satisfying $h|_{\eom[P']}\notin \mathscr{P}$. Therefore we have:
    $$\ceo(\{g|_{\eom[P]},h|_{\eom[P']}\})\le_T\ceo(\mathcal{F}\cup\{\ouhe\}).$$
    Let $g'\in\tau(g|_{\eom[P]})$ be a signature of arity $k$, and we have $g'
    \notin \mathscr{A}$ by Lemma \ref{lem:tau  pi maintain A P}. Similarly let $h'\in\tau(h|_{\eom[P']})$ be a signature of arity $k$ and not in $\mathscr{P}$.

    By Lemma \ref{thm:csp=eom}, $\#\csp(\{g',h'\})\equiv_T \ceo(\{\pi(g'),\pi(h')\})$, and since we can rename the variables in the $\#\eo$ framework, we have:
    $$\#\csp(\{g',h'\})\leq_T\ceo(\{g|_{\eom[P]},h|_{\eom[P']}\})\leq_T\ceo(\mathcal{F}\cup\{\ouhe\}).$$
    By Theorem \ref{thm:CSPdichotomy} we have $\ceo(\mathcal{F}\cup\{\ouhe\})$ is \#P-hard.
   
\end{proof}
Now we prove that the $\ouhe$ signature can be constructed in Lemma \ref{3uphard}. We generalize this result as Lemma \ref{lem:obtainouhe}.
\begin{lemma}
    Suppose $f\in\mathcal{F}$ is a $\mitsuup$ signature, but not a $\allup$ signature. Then $$\ceo(\mathcal{F}\cup\{\ouhe\})\le_T\ceoc(\mathcal{F}).$$
    \label{lem:obtainouhe}
\end{lemma}
We introduce some more definitions before proving Lemma \ref{lem:obtainouhe}. For $\alpha,\beta\in\eoe$ and $\dis\alpha\beta=2s$, where $s\in\mathbb{N}_+$, we say $\gamma\in\eoe$ is \textit{between} $\alpha$ and $\beta$ if there exists a positive integer $t<s$ such that $\dis\alpha\gamma=2t$ and $\dis\beta\gamma=2s-2t$. In other words, by reordering some indices, $\alpha=\sigma\tau\pi,\beta=\sigma\overline{\tau\pi}$ and $\gamma=\sigma\tau\overline{\pi}$, where $|\pi|=2t,|\tau|=2s-2t$.

\begin{lemma}
    Suppose $f\in\mathcal{F}$ is a $\mitsuup$ signature, and $\ceo(\mathcal{F}\cup\{\ouhe\})\le_T\ceoc(\mathcal{F})$ does not hold. Then for any $\alpha=\sigma\tau,\beta=\sigma\overline{\tau}\in\su(f)$ (variables reordered), we have $\su(f^{x_1\cdots x_{\len\sigma}=\sigma})=\eo^P$ for some $P$, where $P$ is a pairing of variables of $f^{x_1\cdots x_{|\sigma|}=\sigma}$. 
    \label{lem:noouhe}
\end{lemma}
\begin{proof}
    The lemma is trivial when $\dis\alpha\beta=2$. We begin with proving several statements.
    
    Firstly, we prove that if $\gamma,\delta,\epsilon\in\su(f)$ are between $\alpha$ and $\beta$ or equal to one of them, then $\mu=\gamma\oplus\delta\oplus\epsilon\in \eoe$. Let $\gamma'=\gamma\oplus\alpha\oplus\beta,\delta'=\delta\oplus\alpha\oplus\beta,\epsilon'=\epsilon\oplus\alpha\oplus\beta$ and $\mu'=\gamma'\oplus\delta'\oplus\epsilon'$, then $\gamma',\delta',\epsilon'\in\eoe$. As $f$ is a $\mitsuup$ signature, we also have  $\gamma',\delta',\epsilon'\in \su(f)$ and consequently both $\mu,\mu'\in\eog$. If $\mu\in \eosg$, then $\mu'\in\eosl$, which is a contradiction. Therefore, $\mu\in\eoe$.
    
    Secondly, if $\dis\alpha\beta\ge4$, then there must exist a $\gamma\in\su(f)$ between $\alpha$ and $\beta$. Otherwise, $f^\sigma$ would be equal to $\neq_{\dis\alpha\beta}^{a,b}, ab\neq0$. Without loss of generality, we may assume $\su(\neq_{\dis\alpha\beta}^{a,b})=\{01\dots01,10\dots10\}$. By adding a self-loop on $x_{2k+3},x_{2k+4}$ for each $1\le k\le \dis\alpha\beta/2-2$, we realize $\neq_4^{a,b}$, by Lemma \ref{lem:neq4ab obtain ouhe} we may further realize $\ouhe$, which causes a contradiction. 
    
    By Lemma \ref{lem: eo+affine= eom}, the first statement implies that $\su(f^{x_1\cdots x_{\len\sigma}=\sigma})$ is affine, thus it is a subset of $\eom[P]$ for some pairing $P$ of variables of $f^{x_1\cdots x_{\len\sigma}=\sigma}$. We now use the induction to prove this lemma. When $\dis\alpha\beta=2$, it's obvious that $\su(f^\sigma)=\eo^P$. Assume that this lemma is correct when $\dis\alpha\beta<2k$, where $k\geq2$. Now suppose $\dis\alpha\beta=2k$. By repeatedly using the second statement, there would be a $\gamma\in\su(f^\sigma)$ between $\alpha$ and $\beta$ satisfying $\dis\alpha\gamma=2$. Suppose $\len\sigma=L$. After reordering variables, we may let $\alpha=01\sigma\tau',\beta=10\sigma\overline{\tau'}$ and $\gamma=10\sigma\tau'$. By the induction hypothesis, there exists a pairing $P'$ such that $\su(f^{x_1\cdots x_{L+2}=10\sigma})=\eo^{P'}$. Let $P=P'\cup\{(x_1,x_2)\}$ be a pairing of $f^{x_3\cdots x_{L+2}=\sigma}$. For any $\mu\in \eo^{P'}$, $\delta=10\sigma\mu\in \su(f)$, and $\epsilon=\delta\oplus\alpha\oplus\gamma=01\sigma\mu\in\su(f)$. This implies $\eo^P\subseteq\su(f^{x_3\cdots x_{L+2}=\sigma})$. Since $\su(f^{x_3\cdots x_{L+2}=\sigma})\subseteq\eo^Q$ for some pairing $Q$, we have $|\su(f^{x_3\cdots x_{L+2}=\sigma})|=|\eom[Q]|=|\eom[P]|$. Consequently, $\su(f^\sigma)=\eom[P]$. The induction is done.
\end{proof}

Now we are ready to prove Lemma \ref{lem:obtainouhe}.

\begin{proof}
    Suppose otherwise. As $f$ is not a $\allup$ signature, there exists a minimal odd integer $k$ such that there exist $k$ strings $\alpha_1,...,\alpha_k\in \su(f)$ satisfying $\delta=\bigoplus_{1\le i\le k}\alpha_i\in \eosl$. 
    Without loss of generality, let $\alpha_1=\sigma\tau,\alpha_2=\sigma\overline{\tau}$ and assume that $|\sigma|=s,|\tau|=t$. By Lemma \ref{lem:noouhe}, there exists a pairing $P$ such that $\su(f^{x_1\cdots x_s=\sigma})=\eom[P]$. Assume that $(x_{s+1},x_{s+2})\in P$ and we let $\alpha_1=\sigma01\tau',\alpha_2=\sigma10\overline{\tau'}$. Then by Lemma \ref{lem:noouhe} we have $\beta=\sigma10\tau'\in \su(f)$, and $\oplus\beta\oplus\alpha_1$ can be seen as an operation that flips the two bits at $x_{s+1},x_{s+2}$. For any $1\le l\le k$, $\alpha_l\oplus\beta\oplus\alpha_1\in \eog$, so $x_{s+1}x_{s+2}$ cannot be $11$ at $\alpha_l$.

    Let $\gamma=\bigoplus_{3\le i\le k}\alpha_i\in\eog$. Then we can write $\delta=\gamma\oplus\alpha_1\oplus\alpha_2\in \eosl$. Noticing $\oplus\alpha_1\oplus\alpha_2$ can be seen as an operation that flips all the last $t$ bits, therefore there are more $1'$s in the last $t$ bits of $\gamma$. By the pigeonhole principle, there exists a pair $(x_p,x_q)$ in $P$ such that $\gamma_{p}\gamma_{q}$ is $11$. This also implies that $(\delta)_p=(\delta)_q=0$. Without loss of generality, we still assume this pair $(x_p,x_q)$ to be $(x_{s+1},x_{s+2})$.
    
    Now we apply the following operation on $\alpha_i$ for each $1\le i\le k$, and in the following we use $(\alpha)_i$ to denote the $i$th bit of a string $\alpha$: If $(\alpha_i)_{s+1}(\alpha_i)_{s+2}=10$, we replace $\alpha_i$ with $\alpha'_i=\alpha_i\oplus\beta\oplus\alpha_1$ (i.e. flip the two bits at $x_{s+1},x_{s+2}$); otherwise we keep $\alpha_i$ unchanged. 
    
    We denote all these operations together as a 0-moving step. After this 0-moving step, all the obtained strings, which we denoted as $\alpha'_1,...,\alpha'_k$, still belong to $\su(f)$. Furthermore, let $\delta'=\bigoplus_{1\le i\le k}\alpha'_i$, by definition we have $(\delta')_i=(\delta)_i$ for all $i\neq s+1,s+2$.
    Since $(\alpha'_i)_{s+1}=0$ holds for each $1\le i\le k$ and $(\delta)_{s+1}(\delta)_{s+2}=00$, $(\delta')_{s+1}(\delta')_{s+2}=00$ as well. This gives $\delta'=\delta\in\eosl$, meaning that this property still holds after the 0-moving step.  
    
    Now we repeatedly apply the 0-moving step on $\alpha_1,...,\alpha_k$. Suppose the arity of $f$ is $n$. For a variable $x_s,1\le s\le n$, we say $x_s$ is all-0 if $(\alpha_i)_s=0$ for all $1\le i \le k$. Notice that in each of such steps, the number of all-0 variables increases exactly by 1. 
    This means at most $n$ steps can be applied. However on the other hand, the 0-moving step can be applied as long as there are 2 different strings among the $k$ strings we currently choose, implying that the step can be applied for unlimited times. This contradiction finishes our proof.
\end{proof}

\section{Extensions of the dichotomy}\label{sec:extend}
Although this article focuses on \ceo\ problems, the dichotomy can actually be extended to a wider range of signature sets whose complexity was previously unknown. In this section, we introduce two related dichotomies which can be implied from the dichotomy for \ceo.

\subsection{Dichotomy for \upside/ \downside\ signatures}
In this section we deal with the dichotomy for \upside/\downside\ signatures. Suppose $f$ is a signature of arity $2k-1$ or $2k$ ($k\in \mathbb{N}^+$). If $\su(f)\subseteq\eog$, we say $f$ is an \textit{\upside}\ signature; if $\su(f)\subseteq\eosg$, we say $f$ is a \textit{\supside}\ signature. Similarly, if $\su(f)\subseteq\eol$, we say $f$ is a \textit{\downside}\ signature; if $\su(f)\subseteq\eosl$, we say $f$ is a \textit{\sdownside}\ signature. 

For each signature $f\in \mathcal{F}$ of arity $k$ and $\alpha\in\{0,1\}^k$, let 
    \begin{equation}
        f|_{\eo}(\alpha)=\begin{cases}
            f(\alpha), & \alpha\in\eoe;\\
            0, & \alpha\notin \eoe.\notag
        \end{cases}
    \end{equation}
    and $\mathcal{F}|_{\eo}=\{f|_{\eo}|f\in\mathcal{F}\}$. We remark that for a signature $f$ of odd arity, $f|_{\eo}$ remains constant at 0 by this definition.

Then we have the following corollary:
\begin{corollary}
    Suppose $\mathcal{F}$ is a set of \upside\ (or \downside\ respectively) signatures.  Then $\hol(\neq_2\mid\mathcal{F})$ is \#P-hard, unless all signatures in $\mathcal{F}|_{\eo}$ are $\mitsuup$ signatures or all signatures in $\mathcal{F}|_{\eo}$ are $\mitsudown$ signatures, and $\mathcal{F}|_{\eo}\subseteq \eom[\mathscr{A}]$ or $\mathcal{F}|_{\eo}\subseteq \eom[\mathscr{P}]$, in which cases it is in $\text{FP}^{\text{NP}}$.
    \label{coro:upside}
\end{corollary}
The proof use a similar idea as Lemma \ref{lem:3upeff}.
\begin{proof}
    We assume that $\mathcal{F}$ is a set of \upside\ signatures. The proof is similar when $\mathcal{F}$ is a set of \downside\ signatures.

    Given any instance $I$, we first show that any support $\delta_f\in\eosg$ of a signature $f\in\mathcal{F}$ cannot be effective. Suppose otherwise and $\delta=\delta_f\delta'$ satisfies that $\omega_I(\delta)\neq 0$.
    Each variable in $I$ is incident to a signature in $\mathcal{F}$ on RHS and a $\neq_2$ signature on LHS, thus both LHS and RHS cover each variable exactly once. For each signature $g\in \mathcal{F}$ on RHS, $\su(g)\subseteq\eog$, thus $\delta'\in\eog$. As $\delta_f\in\eosg$, we also have $\delta\in\eosg$. However, by the restriction from LHS, we have $\delta\in \eoe$, which is a contradiction.

    By the analysis above we have $\hol(\neq_2\mid\mathcal{F})\equiv_T \hol(\neq_2\mid\mathcal{F}|_{\eo})$. Since $\mathcal{F}|_{\eo}$ only consists of $\eo$ signatures and signatures of odd arity remain constant at 0, by Theorem \ref{thm:dicoceo} we are done.
\end{proof}
\subsection{Dichotomy for \sw\ signatures}
In this section, we deal with \sw\ signatures. Suppose $f$ is a signature of arity $k$. If $f(\alpha)$ only takes nonzero values at Hamming weight $d,0\le d\le k$, we say $f$ is a \textit{\sw}\ signature. By definition, $f$ is an \eo\ signature if $2d=k$, a strictly \upside\ signature if $2d>k$, and a strictly \downside\ signature if $2d< k$. 

For each single-weighted signature $f\in \mathcal{F}$ of arity $k$ which only takes nonzero values at Hamming weight $d,0\le d\le k$, let 
    \begin{equation}
        f_{\to\eo}=\begin{cases}
            f\otimes \Delta_0^{2d-k} & 2d\ge k;\\
            f\otimes \Delta_1^{k-2d}, & 2d<k.\notag
        \end{cases}
    \end{equation}
    Let $\mathcal{F}_{\to\eo}=\{f_{\to\eo}|f\in\mathcal{F}\}$. We remark that for an \eo\ signature $f$, $f_{\to\eo}=f$. We also remark that each signature in $\mathcal{F}_{\to\eo}$ is an \eo\ signature.

Then we have the following corollary:
\begin{corollary}
    Suppose $\mathcal{F}$ is a set of \sw\ signatures. Then $\hol(\neq_2\mid\mathcal{F})$ is \#P-hard, unless one of the following holds, in which cases it is in $\text{FP}^{\text{NP}}$.
    \begin{enumerate}
        \item All signatures in $\mathcal{F}$ are $\eog$ (or $\eol$ respectively) signatures.  In addition, all signatures in $\mathcal{F}|_{\eo}$ are $\mitsuup$ signatures or all signatures in $\mathcal{F}|_{\eo}$ are $\mitsudown$ signatures, and $\mathcal{F}|_{\eo}\subseteq \eom[\mathscr{A}]$ or $\mathcal{F}|_{\eo}\subseteq \eom[\mathscr{P}]$;
        \item There exist a signature that is not $\eog$ and a signature that is not $\eol$ belonging to $\mathcal{F}$.  In addition, All signatures in $\mathcal{F}_{\to\eo}$ are $\mitsuup$ signatures or all signatures in $\mathcal{F}_{\to\eo}$ are $\mitsudown$ signatures, and $\mathcal{F}_{\to\eo}\subseteq \eom[\mathscr{A}]$ or $\mathcal{F}_{\to\eo}\subseteq \eom[\mathscr{P}]$.
    \end{enumerate}
    \label{coro:sw}
\end{corollary}

\begin{proof}
    If all signatures in $\mathcal{F}$ are \upside\ (or \downside\ respectively) signatures, we are done by Corollary \ref{coro:upside}. Consequently, we may assume there exist an $f\in \mathcal{F}$ which is not an \upside\ signature and a $g\in \mathcal{F}$ which is not a \downside\ signature. The proof consists of two parts. The first part realize the pinning signatures on the left side, while the second part reduce the problem to a $\ceoc$ problem with the help of the pinning signatures. 

     \subparagraph{Obtaining $\pin_0,\pin_1$ signatures} \par In this part, we first prove the statement $\hol(\neq_2\mid\mathcal{F})\equiv_T \hol(\neq_2\mid\mathcal{F}\cup\{\pin_0\})$. Suppose $f$ is of arity $k$. As $f$ is \sw\ and not \upside, $f$ can only be a \sdownside\ signature and 
     there exists an $\alpha\in\su(f)\cap\eosl$. We prove this statement by induction on $\#_1(\alpha)$. If $\#_1(\alpha)=0$, $f=\pin_0^{\otimes k}$, and we are done by Lemma \ref{lem:linwang}. 
     
     Suppose the statement holds when $\#_1(\alpha)=d\ge 0$, and we prove the statement for the case $\#_1(\alpha)=d+1$. As $d\ge 0$, $\#_1(\alpha)\ge 1$ and $\#_0(\alpha)\ge 2$. Without loss of generality, we may assume $\alpha=100\delta$. We also let $\beta=010\delta,\gamma=001\delta$. If $f(\alpha)+f(\beta)\neq 0$, by adding a self-loop by $\neq_2$ to $x_1,x_2$, we obtain the signature $f_{12}$ satisfying $f_{12}(0\delta)=f(\alpha)+f(\beta)\neq 0$. Since $\#_1(0\delta)=d$, we are done by induction. Similarly, if $f(\alpha)+f(\gamma)\neq 0$ or $f(\beta)+f(\gamma)\neq 0$, by adding a self-loop by $\neq_2$ to $x_1,x_3$ or $x_2,x_3$ we are done as well. As $f(\alpha)\neq 0$, at least one of $\{[f(\alpha)+f(\beta)\neq 0],[f(\alpha)+f(\gamma)\neq 0],[f(\beta)+f(\gamma)\neq 0]\}$ holds, and consequently the statement is proved.

     Similarly, with the signature $g$ we can also obtain the $\pin_1$ signature. Consequently, we have
     
     $$\hol(\neq_2\mid\mathcal{F})\equiv_T \hol(\neq_2\mid\mathcal{F}\cup\{\pin_0,\pin_1\})$$

     \subparagraph{Reducing from $\ceoc$} 
     In this part, we prove that 
     
     $$\hol(\neq_2\mid\mathcal{F}\cup\{\pin_0,\pin_1\})\equiv_T \hol(\neq_2\mid\mathcal{F}_{\to\eo}\cup\{\pin\})$$
     
     Each signature in $\mathcal{F}_{\to\eo}\cup\{\pin\}$ can be constructed from signatures in $\mathcal{F}\cup\{\pin_0,\pin_1\}$  by tensor product, thus $\hol(\neq_2\mid\mathcal{F}_{\to\eo}\cup\{\pin\})\le_T \hol(\neq_2\mid\mathcal{F}\cup\{\pin_0,\pin_1\})$. On the other hand, for each instance $I$ of $\hol(\neq_2\mid\mathcal{F}\cup\{\pin_0,\pin_1\})$, if there exists a $\delta$ such that $\omega_I(\delta)\neq 0$, then $\#_0(\delta)=\#_1(\delta)$ by the restriction from LHS. As each signature from RHS is a \sw\ signature, the values of $\#_0(\delta),\#_1(\delta)$ are also decided by the restriction from RHS. If $\#_0(\delta)\neq\#_1(\delta)$ by the restriction from RHS, it would be a contradiction and for each possible $\delta$ we have $\omega_I(\delta)=0$ and consequently $\text{Z}(I)=0$, making the reduction trivial.

     Otherwise, the tensor product of all signatures from RHS in $I$ form an \eo\ signature. We construct the instance $I'$ of $\hol(\neq_2\mid\mathcal{F}_{\to\eo}\cup\{\pin\})$ in the following way. Here, by replacing a signature $f$ with $f'$, we mean that we increase the arity of $f$ by tensoring a certain number of $\pin_0$ or $\pin_1$, while keep the origin connections that $f$ involves unchanged.
     \begin{enumerate}
         \item Replacing each $\pin_0$ in $I$ with $\pin$, and each $\pin_1$ with $\pin$ as well. 
          \item  If $f$ is an \upside\ signature and there exists an $\alpha\in\su(f)$ satisfying $\#_1(\alpha)-\#_0(\alpha)=a$, replace $f$ with $f\otimes \pin_0^{\otimes a}$ from $\mathcal{F}_{\to\eo}$.
         \item  If $f$ is a \downside\ signature and there exists a $\beta\in\su(f)$ satisfying $\#_0(\beta)-\#_1(\beta)=b$, replace $f$ with $f\otimes \pin_1^{\otimes b}$ from $\mathcal{F}_{\to\eo}$. 
         \item For each introduced $\pin_0$, connect it to an introduced $\pin_1$ via a $\neq_2$ signature on LHS.
     \end{enumerate}

    Each signature in $\mathcal{F}_{\to\eo}\cup\{\pin\}$ is an \eo\ signature, and consequently the tensor product of all signatures from the left side in $I'$ also forms an \eo\ signature. This means that when constructing $I'$, the number of $\pin_0$ introduced equals that of $\pin_1$  introduced, and consequently the fourth step in the construction leaves no dangling signatures. Also notice that in the construction, we only introduce several connecting components formed by a $\pin_0$ and a $\pin_1$, and each of them has no effect to the partition function since the value of each equals 1. Consequently, $\text{Z}(I)=\text{Z}(I')$ and we have $\hol(\neq_2\mid\mathcal{F}\cup\{\pin_0,\pin_1\})\le_T \hol(\neq_2\mid\mathcal{F}_{\to\eo}\cup\{\pin\})$.

     Combining the two parts, we prove that $\hol(\neq_2\mid\mathcal{F})\equiv_T \hol(\neq_2\mid\mathcal{F}\cup\{\pin_0,\pin_1\})\equiv_T \ceoc(\mathcal{F}_{\to\eo})$, and by Theorem \ref{thm:dicoceoc} we are done.
\end{proof}

\section{Conclusion and discussion}\label{sec:conclusion}

In this article, we prove an $\pnp$ vs. \#P dichotomy for complex \ceo\ over Boolean domain, and extend it to dichotomies for \upside, \downside\ and \sw\ signatures respectively. 

We remark that for most of the cases in this article, we have already achieved an FP vs. \#P dichotomy. The only exception is that all signatures in $\mathcal{F}$ are $\mitsuup$ signatures or all signatures in $\mathcal{F}$ are $\mitsudown$ signatures, and $\mathcal{F}\subseteq \eom[\mathscr{A}]$ or $\mathcal{F}\subseteq \eom[\mathscr{P}]$ while $\mathcal{F}$ is neither \ba[0] nor \ba[1]. We remark that such case does exist. The \eo\ signature $f_{56}$, introduced in \cite{meng2024p}, is a $\mitsuup$ $\eom[\mathscr{P}]$ signature, and does not satisfy the \ba[0] or \ba[1] property. In order to eliminate the discrepancy between the $\pnp$ vs. \#P dichotomy and an FP vs. \#P dichotomy, it is necessary to establish a complexity classification for the support identification problem. A decision problem of this type is referred to as a Boolean constraint satisfaction problem. This specific area of study has been the subject of investigation by \cite{feder2006classification}. Nevertheless, a comprehensive complexity classification remains to be established, and the issue addressed in this article is not yet encompassed by extant research. 

Seeking a complete dichotomy for \hol\  is also a worthwhile pursuit. Despite the fact that this problem has been open for over fifteen years, we hope that the results presented in this paper will contribute to the advancement of the field. 


\bibliography{ref}

\appendix

\end{document}